\documentclass[journal,onecolumn]{IEEEtran}
\newcommand{\new}[1]{#1}

\IEEEoverridecommandlockouts
\usepackage{cite}
\usepackage{amsmath,amssymb,amsfonts}
\usepackage{mathrsfs}
\usepackage{amsthm}
\usepackage{algorithm}

\usepackage[noend]{algpseudocode}
\usepackage{multicol}
\usepackage{stmaryrd}
\usepackage{graphicx}
\usepackage{textcomp}
\usepackage{xcolor}
\usepackage{changes}
\usepackage{ulem}
\usepackage{bbm}
\usepackage{lipsum,adjustbox}
\usepackage{pgfplots}
\usepackage{setspace}
\usepackage{hyperref}
\hypersetup{
 colorlinks,
 linkcolor={blue!100!black},
 citecolor={blue!100!black},
 urlcolor={blue!80!black}
}
\usepackage[numbers,sort]{natbib}
\newtheorem{theorem}{Theorem}

\newtheorem{example}{Example}
\newtheorem{lemma}{Lemma}

\newtheorem{corollary}{Corollary}
\newtheorem{condition}{Condition}

\theoremstyle{definition}
\newtheorem{definition}{Definition}
\newtheorem{remark}{Remark}

\usepackage{soul} 
\newcommand{\green}[1]{\textcolor{green}{#1}}

\newcommand{\blue}[1]{\textcolor{blue}{#1}}
\newcommand{\red}[1]{\textcolor{red}{#1}}
\newcommand{\gray}[1]{\textcolor{lightgray}{#1}}
\newcommand{\orange}[1]{\textcolor{orange}{#1}}

\newcommand{\etal}{\textit{et al. }}

\newcommand{\FUNC}[1]{\textsc{\small{#1}}}
\newcommand{\Var}[1]{\textit{#1}}

\newcommand{\coded}[1]{\widetilde{#1}}

\newcommand{\sig}[1]{\langle#1\rangle}

\algblockdefx{Upon}{EndUpon}[1]{\textbf{wait for} #1}

\algblockdefx{As}{EndAs}[1]{\textbf{as} #1}

\algrenewcommand\algorithmicindent{.8em}%
\algrenewcommand\algorithmiccomment[1]{\gray{\hfill // #1}}
\newcommand{\CommentX}[1]{\gray{// #1}}

\makeatletter
\ifthenelse{\equal{\ALG@noend}{t}}%
  {\algtext*{EndAs}}
  {}%
\makeatother
\makeatletter
\ifthenelse{\equal{\ALG@noend}{t}}%
  {\algtext*{EndUpon}}
  {}%
\makeatother

\newcommand{\bbF}{\mathbb{F}}

\newcommand{\bfa}{\mathbf{a}}
\newcommand{\bfb}{\mathbf{b}}
\newcommand{\bfc}{\mathbf{c}}

\newcommand{\bfe}{\mathbf{e}}

\newcommand{\bfh}{\mathbf{h}}

\newcommand{\bfm}{\mathbf{m}}

\newcommand{\bfp}{\mathbf{p}}

\newcommand{\bfs}{\mathbf{s}}

\newcommand{\bfu}{\mathbf{u}}
\newcommand{\bfv}{\mathbf{v}}
\newcommand{\bfw}{\mathbf{w}}
\newcommand{\bfx}{\mathbf{x}}
\newcommand{\bfy}{\mathbf{y}}

\newcommand{\bfB}{\mathbf{B}}
\newcommand{\bfC}{\mathbf{C}}

\newcommand{\bfF}{\mathbf{F}}
\newcommand{\bfG}{\mathbf{G}}

\newcommand{\bfM}{\mathbf{M}}

\newcommand{\bfR}{\mathbf{R}}

\newcommand{\bfU}{\mathbf{U}}
\newcommand{\bfV}{\mathbf{V}}
\newcommand{\bfW}{\mathbf{W}}

\newcommand{\cA}{\mathcal{A}}

\newcommand{\cC}{\mathcal{C}}

\newcommand{\cF}{\mathcal{F}}
\newcommand{\cG}{\mathcal{G}}

\newcommand{\cI}{\mathcal{I}}
\newcommand{\cJ}{\mathcal{J}}
\newcommand{\cK}{\mathcal{K}}
\newcommand{\cL}{\mathcal{L}}
\newcommand{\cM}{\mathcal{M}}

\newcommand{\cS}{\mathcal{S}}

\newcommand{\itf}{\textit{f}}

\newcommand{\ith}{\textit{h}}

\newcommand{\itm}{\textit{m}}

\newcommand{\itr}{\textit{r}}

\newcommand{\itt}{\textit{t}}

\newcommand{\itv}{\textit{v}}

\newcommand{\header}{\textit{header}}
\newcommand{\type}{\textit{type}}
\newcommand{\payload}{\textit{payload}}
\newcommand{\fragment}{\textit{fragment}}
\newcommand{\highQC}{\textit{highQC}}
\newcommand{\prepareQC}{\textit{prepareQC}}
\newcommand{\precommitQC}{\textit{precommitQC}}
\newcommand{\commitQC}{\textit{commitQC}}
\newcommand{\lockedQC}{\textit{lockedQC}}

\newcommand{\viewNumber}{\textit{viewNumber}}
\newcommand{\curView}{\textit{curView}}
\newcommand{\qc}{\textit{qc}}

\newcommand{\checksum}{\textit{checksum}}
\newcommand{\checksums}{\textit{checksums}}
\newcommand{\prev}{\textit{prev}}
\newcommand{\cks}{\textit{cks}}
\newcommand{\ack}{\textit{ack}}
\newcommand{\signature}{\textit{sig}}
\newcommand{\partialSig}{\textit{partialSig}}
\newcommand{\px}{\Var{px}}

\newcommand{\ralg}[2]{line~#1, Algorithm~#2}
\newcommand{\ralgx}[3]{line~#1--#2, Algorithm~#3}

\newcommand\blfootnote[1]{%
  \begingroup
  \renewcommand\thefootnote{}\footnote{#1}%
  \addtocounter{footnote}{-1}%
  \endgroup
}

    \makeatletter
    \renewcommand{\ALG@beginalgorithmic}{\footnotesize}
    \makeatother

\begin{document}

\title{Breaking Blockchain's Communication\\Barrier with Coded Computation}

\author{\textbf{Canran Wang} and \textbf{Netanel Raviv}\\Department of Computer Science and Engineering, Washington University in St. Louis, St. Louis, MO 63103.}

\maketitle
\thispagestyle{plain}
\pagestyle{plain}

\begin{abstract}
\blfootnote{Parts of this paper have previously appeared in~\cite{MyPaper}.}
Although blockchain, the supporting technology of various cryptocurrencies, has offered a potentially effective framework for numerous decentralized trust management systems, its performance is still sub-optimal in real-world networks. With limited bandwidth, the communication complexity for nodes to process a block scales with the growing network size and hence becomes the limiting factor of blockchain's performance.

In this paper, we suggest a re-design of existing blockchain systems, which addresses the issue of the communication burden. First, by employing techniques from Coded Computation, our scheme guarantees correct verification of transactions while reducing the bit complexity dramatically such that it grows logarithmically with the number of nodes. Second, with the adoption of techniques from Information Dispersal and State Machine Replication, the system is resilient to Byzantine faults and achieves linear message complexity. Third, we propose a novel 2-dimensional sharding strategy, which inherently supports cross-shard transactions, alleviating the need for complicated communication protocols between shards, while keeping the computation and storage benefits of sharding. 
\end{abstract}

\section{Introduction}\label{section:introduction}

Blockchain is an append-only decentralized system, in which data resides in a chain of blocks that are periodically proposed and agreed upon by a consensus mechanism.
Although it is a promising platform for various applications, its performance is sub-optimal due to the limited bandwidth and the scaling communication complexity, in terms of~\emph{message complexity} and~\emph{bit complexity}. 
Message complexity is measured by the number of transferred messages, and bit complexity is characterized by the number of communicated bits. 

The performance of Bitcoin~\cite{Bitcoin}, the first double-spending-resistent cryptocurrency in a public peer-to-peer network, is inherently limited by its design.
For a valid new block to be generated, the competing nodes invest the majority of time in solving Proof-of-Work (PoW) puzzles. 
Consensus is reached on the sole block proposed by the winner, which is then propagated to the remaining nodes and appended to each local chain. 
Such a concatenated consensus-then-propagation scheme fails to fully utilize the bandwidth of nodes, since the network remains idle during the PoW puzzle solving period. 

Meanwhile, the security of Bitcoin is guaranteed by the fact that the time interval between blocks is sufficiently greater than the block transmission time~\cite{OnScaling}. Otherwise, frequent forks, which occur when multiple blocks are proposed simultaneously cause temporary inconsistency between nodes, and jeopardize the safety of the system. In other words, the PoW puzzle should take a sufficiently long period of time to solve, compared with the required time for the majority of node to receive a block. Together, na\"{i}ve reparameterization such as reducing the difficulty of the PoW puzzle or enlarging the block size degrades security, and a comprehensive redesign is required to improve Bitcoin's performance.

A widely adopted paradigm to achieve high-performance blockchain systems is to parallelize consensus and propagation, and hence to maximize the efficiency of bandwidth usage~\cite{BitcoinNG,HybridConsensus,ByzCoin,Prism}. Works following this path inherit the PoW mechanism to periodically select an entity from the public network as a leader, which could be a node or a committee of nodes. 
The selected entity is allowed to continuously generate blocks in parallel with the leader election mechanism, until the next entity is selected. 
Compared with Bitcoin, this paradigm persistently utilizes the bandwidth of nodes, and hence improves system performance.

\new{
In another direction, researchers attempt to improve Bitcoin by replacing its PoW mechanism, which is seen as the root cause of the scalability issue and huge energy consumption. 
Proof-of-Stake (PoS) is a noteworthy alternative used by~\cite{Ouroboros,SnowWhite,Algorand}, which does not involve the computation-intensive PoW puzzle solving.  
Instead, the chance for each individual node being selected as the leader, or one of the leaders, is proportional to its~\emph{stake}, referring to the value resides in the blockchain system.
}


Although the aforementioned attempts improve the performance of blockchain to some extent, a fundamental obstacle remains. 
That is, \emph{every node must receive every transaction}. 
This requirement is paramount to the safety and decentralization of blockchain systems, but unfortunately leads to an inevitable~$\Omega(NP)$ bit complexity for a block~$\bfB$ containing~$P$ transactions to be confirmed, given a network of~$N$ nodes.

Sharding~\cite{SokSharding} is a novel paradigm proposed to address this problem. The network is sliced into multiple communities of a similar sizes, each individually processes a disjoint set of data~\cite{Elastico,RapidChain,Omniledger}. The constant community size reduces the communication complexity as one transaction is only propagated within one community. As a result, the system throughput scales with the number of nodes, as additional nodes form extra communities and process additional transactions.

In sharding-based blockchain designs, random node rotation, or even reassignment, is necessary to avoid concentration of adversaries in one community. 
Further, sharding creates a distinction between two types of transactions; a transaction is called intra-shard if the sender and the receiver belong to the same community, and called cross-shard otherwise. 
Hence, extra mechanisms are required in this path, which is an added complexity that degrades the system's performance and diminishes the benefits of sharding.

Coding has been introduced to bypass the requirement for every node to receive every transaction, which leads to the invention of~\emph{coded blockchain}. 
Duan~\etal propose \emph{BEAT3}~\cite{BEAT}, a BFT storage system that enables each node to periodically store a relatively small coded fragment generated from the whole data block. 
The error-correcting code guarantees reconstruction of the original data block from sufficiently many of fragments. 
The AVID-FP~\cite{AVID-FP} protocol is used to assure that fragments stored by correct nodes correspond to a unique original data block. 
However, as a BFT storage system, BEAT3 does not concern \emph{external validity}, which assures that the stored data is acceptable to a specific application~\cite{EV}. 
In Blockchain's scenario, nodes in BEAT3 cannot verify the correctness of each stored transaction.

The introduction of coded computation partially alleviates the security problems in sharding, and provides support for external validity. 
Li~\etal~\cite{Polyshard} proposed \emph{Polyshard}, which offers a novel separation between nodes and shards. Polyshard formulates the verification of transactions as computation tasks, one for each shard, to be solved across all nodes in a distributed manner. 
Using Lagrange Coded Computing (LCC)~\cite{LCC}, nodes individually compute a polynomial verification function over a \emph{coded chain} and a \emph{coded block}. 
Since verification is performed in a coded fashion, and a node does not verify or store transactions for any specific shard, the need for node rotation/reassignment is removed.

Polyshard implicitly assumes that the performance bottleneck stems from insufficient computation resources in nodes, rather than limited communication bandwidth, and considers the system as a computation cluster with a highly synchronous network. 
The bit complexity, however, is again~$O(NP)$, as Polyshard requires every node to firstly reach a consensus of the whole block~$\bfB$ and then perform encoding individually. 
\new{Otherwise, as pointed out in~\cite{Discrepancy}, the system can be broken by a discrepancy attack.}
Besides, the messages complexity is~$O(N^2)$, due to the fact that Polyshard involves an all-to-all communication operation.

\new{Finally, coding has been employed in blockchain system that allows~\emph{light nodes}.
Unlike~\emph{full nodes} that validate and store all transactions, light nodes only download the header of each block, and hence addressing the~$O(NP)$ bit complexity. 
The header contains the root of a Merkle tree whose transactions are the leaves; it allows light nodes to verify the inclusion of any transaction in the corresponding block by downloading a Merkle proof from full nodes.
However, the~\emph{data availability problem} arises, i.e., upon receiving a header and Merkle proof, a light node cannot assure the corresponding block is fully available to the network, while the undisclosed part of the block may be invalid.
A coding-based solution to data availability problem has bee proposed by~\cite{DAP} and further improved by~\cite{CMT}.
In this paper, we only consider full nodes, and leave the incorporation of light nodes for future work.}

\subsection*{{Our Contributions}}

In this paper, we propose a fundamental re-design for coded blockchain, which resolves many of the issues in contemporary coded blockchain systems. 
In particular, this re-design addresses the issue that every node should receive every transaction, and hereby resolves the presumably inevitable~$\Omega(NP)$ bit complexity, that is also prevalent in ordinary (that is, uncoded) blockchain systems. 
Further, it achieves linear message complexity by resolving the issue of all-to-all communication, which is message-heavy but necessary for decoding the results of the computation.
On top of this gain, we adopt Lagrange coded computing---similar to existing designs---to achieve comparable levels of decentralization and security guarantees with respect to uncoded (i.e., ordinary) blockchain. In detail,
\begin{enumerate}

\item By employing techniques from Lagrange Coded Computing~\cite{LCC}, our scheme allows nodes to perform verification on~\emph{coded transactions}, whose size is a fraction of the entire block. 
Our method incurs~$O(P\log^2M\log N)$
bit complexity to process a block with~$P$ transactions, where~$M$ is the total number of transactions in one shard.
\new{As an alternative interpretation, the average bit complexity to process a block is~$O(\log^2M\log N)$.}

\item  By devising techniques inspired by Information Dispersal and BFT SMR protocols, our design allows a leader node to securely distribute coded transactions, under the presence of a certain fraction of Byzantine nodes, with~$O(N)$ message complexity in the partial synchrony model. 
\new{In the suggested parameter regime and under standard cryptographic assumptions, our design is provably secure to any attack that aims at breaking the consistency of the system, and in particular the attack pointed out by~\cite{Discrepancy}.
}

\item We propose \emph{2-Dimensional Sharding}, a new technique which partitions the transactions based on their senders and receivers, respectively. This design provides inherent support for cross-shard transactions, alleviating the need for complicated communication mechanisms. 
More precisely, in our design there is \emph{no difference} between the verification process of cross- and intra-shard transactions. 

\item Our design inherits the \emph{unspent transaction output} (UTXO) model and formulates the verification process as computing a polynomial function~\new{with a degree that scales logarithmically with the number of transactions on a shard.} 
In detail, our scheme addresses the degree problem by replacing current cryptographic primitives (i.e, ECDSA, SHA256 and RIPEMD-160) by multivariate cryptography in the generation and verification of a transaction.  

\end{enumerate}

These contributions bring coded blockchain closer to feasibility. 
That is, our scheme achieves linear message complexity and logarithmic bit complexity and removes the boundary between shards with inherent support for cross-shard transactions. 
The rest of this paper is organized as follows. Section~\ref{section:background} introduces necessary background. 
Section~\ref{section:codedVerification} details the coded verification scheme. Section~\ref{section:communication} discusses the communication aspect of our design, including the propagation of transactions and the exchange of computation results. Section~\ref{section:discussion} analyzes the security, communication complexity, and the tradeoff between them. Section~\ref{section:future} discusses the future research directions.

\section{Background}\label{section:background}

\subsection{Lagrange Coded Computing}\label{subsection:LCC}
\new{Coded computing broadly refers to a family of coding-inspired solutions for straggler- and adversary-resilient distributed computation. Tasks of interest include matrix-vector multiplication~\cite{ShortDot}, matrix-matrix multiplication~\cite{PolynomialCode}, gradient-computations~\cite{GradientCoding,GradientCodingCyclic}, and more. 
Further works on the topic include exploitation of partial stragglers~\cite{PartialStragglers}, heterogeneous networks~\cite{HeterogeneousNetworks}, and timely coded computing~\cite{TimelyCodedComputing}.}

Lagrange Coded Computing~\cite{LCC} (LCC) is a recent development in the field of coded computation. 
The task of interest is computing a multivariate polynomial~$f(X)$ on each of the~$K$ datasets~$\{X_1,\ldots,X_K\}$. 
LCC employs the Lagrange polynomial to linearly combine the~$K$ datasets with~$T$ redundant datasets~$\{Z_1,\ldots,Z_T\}$ chosen uniformly at random, generating~$N$ distinct coded dataset~$\{\coded{X}_1,\ldots,\coded{X}_N\}$ with injected computational redundancy. 

The encoding of LCC is performed by first choosing mutually disjoint sets $\{\alpha_1,\ldots,\alpha_N\}$ and $\{\omega_1,\ldots,\omega_K,\ldots,\omega_{K+T}\}$ with elements in~$\bbF_q$. 
The generator matrix~$G_\cL$ is then defined as

\begin{equation}\label{eq:generator matrix}
G_\cL=
\begin{bmatrix}
\Phi_1(\alpha_1) & \Phi_1(\alpha_2) & \ldots &\Phi_1(\alpha_N) \\
\vdots &\vdots&\ddots&\vdots\\
\Phi_{K+T}(\alpha_1) & \Phi_{K+T}(\alpha_2) & \ldots & \Phi_{K+T}(\alpha_N) \\
\end{bmatrix}, 
\end{equation}
where~$\Phi_k(z)$ is the Lagrange polynomial
\begin{align}\label{equation:LagrangePolynomial}
    \Phi_k(z)= \prod_{j,k\in[K+T], j\neq k} \frac{z-\omega_j}{\omega_k - \omega_j}.
\end{align}

The coded datasets is generated as~$(\coded{X}_1,\ldots,\coded{X}_N)=(X_1,\ldots,X_K,Z_1,\ldots,Z_T)\cdot G_\cL.$
Every worker node~$i\in[N]$ computes and returns a coded result~$f(\coded{X}_i)$. The leader obtains $f(X_1),\ldots,f(X_K)$ by performing decoding on collected coded results.

LCC achieves the optimal tradeoff between resiliency, security and privacy. It tolerates up to~$S$ stragglers and~$A$ adversarial nodes, defined as working nodes that are unresponsive or return erroneous results, respectively. In addition, with proper incorporation of random keys, it also prevents the exposure of the original datasets to sets of at most~$T$ colluding workers, as long as 
$$(K+T-1)\deg~f+S+2A+1\leq N.$$

\subsection{State Machine Replication}\label{subsection:smr}
\new{ 
The state machine replication (SMR) approach~\cite{SMR,SMRT} formulates a service, e.g., a network file system, as a state machine to be replicated in participating nodes. 
The state can be altered by client-issued service~\emph{requests}. 
To ensure the consistency of the states, nodes must agree on a total order of execution for requests. 

An SMR implementation must guarantee~\emph{safety} and~\emph{liveness}. 
Safety suggests that no two nodes confirm different order of requests, and~\emph{liveness} imposes that the system continuously accepts and executes new requests.
Further, an SMR protocol is said to be Byzantine Fault-Tolerant (BFT) if it is resilient to Byzantine faults, as coined by Lamport~\etal\cite{ByzGen}, which are defined as arbitrary (and possibly malicious) behaviour of nodes. 

Network models plays an important role in the design of SMR protocols.
In an asynchronous systems, message can be delayed by any finite amount of time, but eventual delivery is guaranteed. 
BFT SMR protocols which operate in this model employ randomization to bypass the famous FLP impossibility~\cite{FLP}. This impossibility result states that in the presence of even one faulty node (not necessarily Byzantine), it is impossible to guarantee consensus with a deterministic protocol. 
Works following this path include~\cite{Honeybadger, Dumbo, BEAT}.

As proposed by Dwork~\etal\cite{DLS}, partial synchrony is another noteworthy network model. 
In this setting, message delivery is asynchronous until an unknown Global Stabilization Time (GST). 
After GST, the system becomes synchronous, where message delay is bounded by a known constant~$\Delta$.
PBFT~\cite{PBFT} is the first practical implementation of BFT SMR in the the partial synchrony model.
It guarantees safety always, and provides liveness when the system becomes synchronous.

PBFT employs a leader to propose client-issued requests, and it takes two phases of all-to-all communication for the decision on one request. 
To prohibit Byzantine leaders from proposing different requests to different nodes, a proposal is considered valid only after being signed by a quorum of~$N-f$ nodes in the first phase, known as a~\emph{quorum certificate (QC)}, where~$f$ is the number of Byzantine nodes. 
Next, nodes commit the request after receiving another~$N-f$ votes in the second phase.
A quorum contains enough nodes such that any two quorums must intersect on at least one correct node~\cite{ByzantineQuorumSystems}.
Such a property guarantees that correct nodes entering the second phase are consistent on the same request.
In addition, it assures that the proposals from subsequent leaders (should the previous one crash) are consistent in request and hence maintains safety during leader switches.
This celebrated two-phase paradigm serves as the foundation of future leader-based BFT SMR protocols~\cite{Zyzzyva, AZyzzyva,SBFT,Tendermint}.

Bitcoin coined the word blockchain, providing an alternative implementation of SMR, particularly for value transfer systems in large networks.
It maintains an ordered sequence (chain) of blocks (requests), each contains transactions that incur value (bitcoins) transfers between clients.
Nodes invest computation power into PoW puzzle solving for the right to propose the next block; they are incentived by a reward in values.
In particular, nodes look for a new block by trial and error. 
The new block must extends the current chain (i.e., contains a hash pointer to the last block on the chain), and the hash value of which must satisfy a certain rule (e.g., begin with a sequence zeros). 
The winner of the competition disseminates its block to the network by gossip protocol, which is then appended to each local chain.
Unlike the protocol discussed earlier, the safety of Bitcoin relies on synchrony. 
Further, the~\emph{finality} property (i.e., a consensus once reached cannot be reverted) of Bitcoin is probabilistic. 
In practice, a block is considered irrevertible after being followed by six new blocks.
Numerous blockchain designs have been introduced to improve Bitcoin (see Section~\ref{section:introduction}).

HotStuff~\cite{Hotstuff} bridges PBFT-like protocols and Bitcoin-like protocols. 
It extends the two-phase paradigm of PBFT to three phases, each contains a nearly identical communication operation between the leader and the nodes.
Due to this remarkable simplicity, HotStuff can be easily pipelined; i.e., the second phase on a block functions as the first phase on the following block, as well as the last phase on the preceding block. 
Therefore, a block is irrevitible after three new blocks being appended to it, which is similar to the case of Bitcoin.
Besides, the extra phase allows HotStuff to be the first protocol that simultaneously achieves linear message complexity and~\emph{responsiveness} during leader switches. 
A leader switch is said to be responsive if the new leader only has to collect a quorum of leader-switch messages, and there is no requirement for it to wait for a predefined time period. 

Due to these merits, we adopt HotStuff as the core consensus protocol, and employ techniques from coded computation and information dispersal (define next) to reduce bit complexity.
Our work can be regarded as a communication-efficient implementation of coded state machine~\cite{CodedStateMachine}, which simultaneously maintains~$K$ state machines (shards) and employ coded computation to combat Byzantine faults.}
\begin{remark}
\new{BFT SMR protocols focus on the communication complexity induced by reaching a consensus on the order of the requests. 
It is generally assumed that each request is broadcast to every node by the issuing client, and this process is out of the scope of communication complexity analysis. 
However, blockchain systems usually require the leader to collect and distribute transactions, which must be considered in analyzing the communication complexity. }
\end{remark}

\subsection{Information Dispersal}\label{subsection:ID}

In a coded distributed information system, a file~$X\in \bbF_q^{\delta K}$ to be stored is first partitioned to~$K$ parts $X=(X_1^\intercal,\ldots,X_K^\intercal)$ where~$X_k^\intercal\in\bbF_q^{\delta\times 1}$. A Maximum Distance Separable (MDS) error-correcting code~$\cC$, induced by a \emph{generator matrix}~${G_\cC\in\bbF_q^{K\times N}}$, is used to generate~$N$ coded fragments $\coded{X}=(\coded{X}_1^\intercal,\ldots,\coded{X}_N^\intercal)=X\cdot G_\cC$. Each of the~$N$ nodes stores one coded fragment. The MDS property of~$\cC$ codes guarantees that any~$K\times K$ submatrix of~$G_\cL$ is of full rank and hence any~$K$ coded fragments are sufficient to reconstruct~$X$, tolerating up to~$N-K$ crashes.

Research in this field normally 
concerns a scenario where an external client wants to \emph{disperse} a file~$X$ to the system. That is, for every node~$i$ to store the corresponding coded fragment~$\coded{X}_i$. 
Byzantine faults can cause inconsistency of coded fragments, i.e., nodes might store coded fragments that do not correspond to the same file~$X$. 
Efforts has been made on developing protocols to combat Byzantine faults in this scenario. 



AVID-FP (where FP stands for \emph{fingerprinting})~\cite{AVID-FP} enables a client to distribute coded fragments of some file~$X$ to nodes in a distributed system, along with a checksum, i.e., a list of~\emph{fingerprints} of every coded fragment. AVID-FP inherits the properties of Cachin's Asynchronous Verifiable Information Dispersal (AVID) protocol~\cite{AVID}, with additional fingerprints. The fingerprints, generated by a homomorphic fingerprinting function (defined formally in the sequel) preserves the structure of error-correcting codes, and allows node~$i$ to verify that the received fragment corresponds to a unique file~$X$. In this paper, we propose an efficient transaction propagation scheme that integrates the steps of AVID-FP and coding techniques (see Section~\ref{section:communication}). 



\subsection{The Unspent Transaction Output (UTXO) Model}
In the UTXO model, value resides in transactions, instead of client accounts. A transaction has inputs and outputs. An unspent output of an old transaction serves as an input to a new transaction, incurring a value transfer between the two. The old UTXO is then invalidated, since it has been spent, and new UTXO is created in the new transaction.

The UTXO model makes extensive use of cryptographic hash functions and digital signatures. The uninformed reader is referred to~\cite[Sec.~2]{Bitcoin} for a thorough introduction to the topic. In a nutshell, a transaction output contains the amount of stored value and the intended receiver's \emph{address}, which is the hash value of her public key. Besides, the sender attaches his public key and signs the transaction with his secret key. 
For a transaction to be valid, the hash value computed from the sender's public key must match the address in the referenced UTXO. Also, the signature must be valid when checked by the public key. This two-step verification process guarantees the sender's possession of the public and secret keys, proves his identity as the receiver of the redeemed UTXO, and protects the integrity of the new transaction.

\new{
Although a transaction may have multiple inputs and outputs, we adopt a simplified UTXO model in our scheme
for clarity, where a transaction has exactly one input and one output, and transfers one indivisible coin. 
}

\subsection{Cryptographic Primitives}\label{subsection:cryptoPrimitives}

We assume that a public key infrastructure (PKI) exists among nodes. That is, every node~$i$ can create a signature~$\sig{m}_{\sigma_i}$ on a message~$m$ using its private key~$\sigma_i$.
Meanwhile, such a signature can be verified by the corresponding public key, which is shared by all nodes. 
Further, we employ a~\emph{threshold signature}~\cite{thresholdSignature} scheme. 
A~$(t,n)$-threshold signature scheme~$\pi$ contains a single public key shared by all nodes. 
Every node~$i$ possess a private key~$\pi_i$ which allows it to create a~\emph{partial signature}~$\sig{m}_{\pi,i}$ on message~$m$. 
A valid threshold signature~$\sig{m}_{\pi}=\Var{tcombine}(m,\{\sig{m}_{\pi,i}\}_{i\in\cI})$ can be produced using function~$\Var{tcombine}$ from a set of partial signatures~$\{\sig{m}_{\pi,i}\}_{i\in\cI}$ of size~$|\cI|=t$, but not smaller. 
Hence, it is guaranteed that the message~$m$ has been signed by~$t$ nodes if the signature verification function~$\Var{tverify}(m, \sig{m}_\pi)$ returns true.

In addition, in order to formulate the verification of transactions as the computation of polynomials, clients use a multivariate public key cryptosystem~(MPKC)~\cite{UOV,Rainbow,Gui} as a signature scheme.
MPKC is based on the multivariate quadratic (MQ) problem, which is believed to be hard even for quantum computers. 
An MQ problem involves a system of~$m$ quadratic polynomials~$\{p^{(1)},\ldots,p^{(m)}\}$ in~$n$ variables~$\{y_1,\ldots,y_n\}$ over some finite field~$\bbF_q$, i.e.,
$$
\bfp(\bfy)= \sum_{0<i\leq j<n}\bfa_{(i,j)}y_iy_j+\sum_{0<i<n}\bfb_iy_i+\bfc,
$$
where~$\bfa$,~$\bfb$, and~$\bfc$ are vectors in~$\bbF_q^m$. The solution is a vector~$\bfu =(u_1,\ldots,u_n)\in\bbF^n$ such that 
$\bfp(\bfu)=(0,\ldots,0)\in\bbF_q^m$.

In general, the public key of a MPKC is the set of coefficients of the quadratic polynomial system. A valid signature~$\bfs\in\bbF_q^n$ on a message~$\bfw\in\bbF_q^m$ is the solution to the quadratic system~$\bfp(\bfy)=\bfw$. In addition to MQ-based signature schemes, hash functions based on multivariate polynomials of low degree have been studied~\cite{LC,MQAnalysis,MQHash}. In the remainder of this paper, we assume a polynomial hash function over~$\bbF_q$ of a constant degree.

\section{Coded Verification}\label{section:codedVerification}

In this section, we first introduce our general settings and assumptions.
\new{
Based on these settings, we discuss the verification of transactions.
As in the UTXO model, the verification process starts from fetching an existing transaction stored in the chain, and proceeds with the address check and signature verification process. 
Together, the entire verification is formulated as computing a polynomial function.}
Consequently, we demonstrate the incorporation of Lagrange Coded Computing, showing how verification can be performed in a coded manner.

\subsection{Setting}

\begin{figure}
    \centering

\tikzset{every picture/.style={line width=0.75pt}} 

\begin{tikzpicture}[x=0.75pt,y=0.75pt,yscale=-1,xscale=1]

\draw  [dash pattern={on 4.5pt off 4.5pt}] (145,30.56) .. controls (145,24.73) and (149.73,20) .. (155.56,20) -- (444.44,20) .. controls (450.27,20) and (455,24.73) .. (455,30.56) -- (455,179.44) .. controls (455,185.27) and (450.27,190) .. (444.44,190) -- (155.56,190) .. controls (149.73,190) and (145,185.27) .. (145,179.44) -- cycle ;
\draw  [fill={rgb, 255:red, 80; green, 227; blue, 194 }  ,fill opacity=0.4 ] (189.28,300) -- (175.93,300) -- (175.93,280) -- (195,280) -- (195,294.28) -- cycle -- (189.28,300) ; \draw   (195,294.28) -- (190.42,295.42) -- (189.28,300) ;
\draw  [fill={rgb, 255:red, 80; green, 227; blue, 194 }  ,fill opacity=0.4 ] (189.28,335) -- (175.93,335) -- (175.93,315) -- (195,315) -- (195,329.28) -- cycle -- (189.28,335) ; \draw   (195,329.28) -- (190.42,330.42) -- (189.28,335) ;
\draw  [fill={rgb, 255:red, 80; green, 227; blue, 194 }  ,fill opacity=0.4 ] (189.28,370) -- (175.93,370) -- (175.93,350) -- (195,350) -- (195,364.28) -- cycle -- (189.28,370) ; \draw   (195,364.28) -- (190.42,365.42) -- (189.28,370) ;
\draw  [fill={rgb, 255:red, 80; green, 227; blue, 194 }  ,fill opacity=0.4 ] (189.28,405) -- (175.93,405) -- (175.93,385) -- (195,385) -- (195,399.28) -- cycle -- (189.28,405) ; \draw   (195,399.28) -- (190.42,400.42) -- (189.28,405) ;
\draw  [fill={rgb, 255:red, 80; green, 227; blue, 194 }  ,fill opacity=0.4 ] (224.28,300) -- (210.93,300) -- (210.93,280) -- (230,280) -- (230,294.28) -- cycle -- (224.28,300) ; \draw   (230,294.28) -- (225.42,295.42) -- (224.28,300) ;
\draw  [fill={rgb, 255:red, 80; green, 227; blue, 194 }  ,fill opacity=0.4 ] (224.28,335) -- (210.93,335) -- (210.93,315) -- (230,315) -- (230,329.28) -- cycle -- (224.28,335) ; \draw   (230,329.28) -- (225.42,330.42) -- (224.28,335) ;
\draw  [fill={rgb, 255:red, 80; green, 227; blue, 194 }  ,fill opacity=0.4 ] (224.28,370) -- (210.93,370) -- (210.93,350) -- (230,350) -- (230,364.28) -- cycle -- (224.28,370) ; \draw   (230,364.28) -- (225.42,365.42) -- (224.28,370) ;
\draw  [fill={rgb, 255:red, 80; green, 227; blue, 194 }  ,fill opacity=0.4 ] (224.28,405) -- (210.93,405) -- (210.93,385) -- (230,385) -- (230,399.28) -- cycle -- (224.28,405) ; \draw   (230,399.28) -- (225.42,400.42) -- (224.28,405) ;
\draw  [fill={rgb, 255:red, 80; green, 227; blue, 194 }  ,fill opacity=0.4 ] (258.35,300) -- (245,300) -- (245,280) -- (264.07,280) -- (264.07,294.28) -- cycle -- (258.35,300) ; \draw   (264.07,294.28) -- (259.5,295.42) -- (258.35,300) ;
\draw  [fill={rgb, 255:red, 80; green, 227; blue, 194 }  ,fill opacity=0.4 ] (258.35,335) -- (245,335) -- (245,315) -- (264.07,315) -- (264.07,329.28) -- cycle -- (258.35,335) ; \draw   (264.07,329.28) -- (259.5,330.42) -- (258.35,335) ;
\draw  [fill={rgb, 255:red, 80; green, 227; blue, 194 }  ,fill opacity=0.4 ] (258.35,370) -- (245,370) -- (245,350) -- (264.07,350) -- (264.07,364.28) -- cycle -- (258.35,370) ; \draw   (264.07,364.28) -- (259.5,365.42) -- (258.35,370) ;
\draw  [fill={rgb, 255:red, 80; green, 227; blue, 194 }  ,fill opacity=0.4 ] (258.35,405) -- (245,405) -- (245,385) -- (264.07,385) -- (264.07,399.28) -- cycle -- (258.35,405) ; \draw   (264.07,399.28) -- (259.5,400.42) -- (258.35,405) ;
\draw  [fill={rgb, 255:red, 80; green, 227; blue, 194 }  ,fill opacity=0.4 ] (293.35,300) -- (280,300) -- (280,280) -- (299.07,280) -- (299.07,294.28) -- cycle -- (293.35,300) ; \draw   (299.07,294.28) -- (294.5,295.42) -- (293.35,300) ;
\draw  [fill={rgb, 255:red, 80; green, 227; blue, 194 }  ,fill opacity=0.4 ] (293.35,335) -- (280,335) -- (280,315) -- (299.07,315) -- (299.07,329.28) -- cycle -- (293.35,335) ; \draw   (299.07,329.28) -- (294.5,330.42) -- (293.35,335) ;
\draw  [fill={rgb, 255:red, 80; green, 227; blue, 194 }  ,fill opacity=0.4 ] (293.35,370) -- (280,370) -- (280,350) -- (299.07,350) -- (299.07,364.28) -- cycle -- (293.35,370) ; \draw   (299.07,364.28) -- (294.5,365.42) -- (293.35,370) ;
\draw  [fill={rgb, 255:red, 80; green, 227; blue, 194 }  ,fill opacity=0.4 ] (293.35,405) -- (280,405) -- (280,385) -- (299.07,385) -- (299.07,399.28) -- cycle -- (293.35,405) ; \draw   (299.07,399.28) -- (294.5,400.42) -- (293.35,405) ;
\draw  [dash pattern={on 4.5pt off 4.5pt}] (145,281.28) .. controls (145,277.96) and (147.69,275.28) .. (151,275.28) -- (309,275.28) .. controls (312.31,275.28) and (315,277.96) .. (315,281.28) -- (315,299.28) .. controls (315,302.59) and (312.31,305.28) .. (309,305.28) -- (151,305.28) .. controls (147.69,305.28) and (145,302.59) .. (145,299.28) -- cycle ;
\draw  [dash pattern={on 4.5pt off 4.5pt}] (171,271) .. controls (171,267.69) and (173.69,265) .. (177,265) -- (195,265) .. controls (198.31,265) and (201,267.69) .. (201,271) -- (201,434) .. controls (201,437.31) and (198.31,440) .. (195,440) -- (177,440) .. controls (173.69,440) and (171,437.31) .. (171,434) -- cycle ;
\draw   (345,275) -- (395,275) -- (395,305) -- (345,305) -- cycle ; \draw   (351.25,275) -- (351.25,305) ; \draw   (388.75,275) -- (388.75,305) ;
\draw   (345,310) -- (395,310) -- (395,340) -- (345,340) -- cycle ; \draw   (351.25,310) -- (351.25,340) ; \draw   (388.75,310) -- (388.75,340) ;
\draw   (345,345) -- (395,345) -- (395,375) -- (345,375) -- cycle ; \draw   (351.25,345) -- (351.25,375) ; \draw   (388.75,345) -- (388.75,375) ;
\draw   (345,380) -- (395,380) -- (395,410) -- (345,410) -- cycle ; \draw   (351.25,380) -- (351.25,410) ; \draw   (388.75,380) -- (388.75,410) ;
\draw [color={rgb, 255:red, 74; green, 144; blue, 226 }  ,draw opacity=1 ]   (420,290) -- (398,290) ;
\draw [shift={(395,290)}, rotate = 360] [fill={rgb, 255:red, 74; green, 144; blue, 226 }  ,fill opacity=1 ][line width=0.08]  [draw opacity=0] (8.93,-4.29) -- (0,0) -- (8.93,4.29) -- cycle    ;
\draw [color={rgb, 255:red, 74; green, 144; blue, 226 }  ,draw opacity=1 ][fill={rgb, 255:red, 74; green, 144; blue, 226 }  ,fill opacity=1 ]   (420,170) -- (420,290) ;
\draw [color={rgb, 255:red, 74; green, 144; blue, 226 }  ,draw opacity=1 ][fill={rgb, 255:red, 74; green, 144; blue, 226 }  ,fill opacity=1 ]   (420,170) -- (390,170) ;
\draw [color={rgb, 255:red, 74; green, 144; blue, 226 }  ,draw opacity=1 ]   (315,290) -- (342,290) ;
\draw [shift={(345,290)}, rotate = 180] [fill={rgb, 255:red, 74; green, 144; blue, 226 }  ,fill opacity=1 ][line width=0.08]  [draw opacity=0] (8.93,-4.29) -- (0,0) -- (8.93,4.29) -- cycle    ;
\draw [color={rgb, 255:red, 126; green, 211; blue, 33 }  ,draw opacity=1 ]   (315,325) -- (342,325) ;
\draw [shift={(345,325)}, rotate = 180] [fill={rgb, 255:red, 126; green, 211; blue, 33 }  ,fill opacity=1 ][line width=0.08]  [draw opacity=0] (8.93,-4.29) -- (0,0) -- (8.93,4.29) -- cycle    ;
\draw [color={rgb, 255:red, 208; green, 2; blue, 27 }  ,draw opacity=1 ]   (315,360) -- (342,360) ;
\draw [shift={(345,360)}, rotate = 180] [fill={rgb, 255:red, 208; green, 2; blue, 27 }  ,fill opacity=1 ][line width=0.08]  [draw opacity=0] (8.93,-4.29) -- (0,0) -- (8.93,4.29) -- cycle    ;
\draw [color={rgb, 255:red, 248; green, 231; blue, 28 }  ,draw opacity=1 ]   (315,395) -- (342,395) ;
\draw [shift={(345,395)}, rotate = 180] [fill={rgb, 255:red, 248; green, 231; blue, 28 }  ,fill opacity=1 ][line width=0.08]  [draw opacity=0] (8.93,-4.29) -- (0,0) -- (8.93,4.29) -- cycle    ;
\draw   (145,221) .. controls (145,217.69) and (147.69,215) .. (151,215) -- (328,215) .. controls (331.31,215) and (334,217.69) .. (334,221) -- (334,239) .. controls (334,242.31) and (331.31,245) .. (328,245) -- (151,245) .. controls (147.69,245) and (145,242.31) .. (145,239) -- cycle ;
\draw  [dash pattern={on 4.5pt off 4.5pt}] (205,271) .. controls (205,267.69) and (207.69,265) .. (211,265) -- (229,265) .. controls (232.31,265) and (235,267.69) .. (235,271) -- (235,434) .. controls (235,437.31) and (232.31,440) .. (229,440) -- (211,440) .. controls (207.69,440) and (205,437.31) .. (205,434) -- cycle ;
\draw  [dash pattern={on 4.5pt off 4.5pt}] (240,271) .. controls (240,267.69) and (242.69,265) .. (246,265) -- (264,265) .. controls (267.31,265) and (270,267.69) .. (270,271) -- (270,434) .. controls (270,437.31) and (267.31,440) .. (264,440) -- (246,440) .. controls (242.69,440) and (240,437.31) .. (240,434) -- cycle ;
\draw  [dash pattern={on 4.5pt off 4.5pt}] (275,271) .. controls (275,267.69) and (277.69,265) .. (281,265) -- (299,265) .. controls (302.31,265) and (305,267.69) .. (305,271) -- (305,434) .. controls (305,437.31) and (302.31,440) .. (299,440) -- (281,440) .. controls (277.69,440) and (275,437.31) .. (275,434) -- cycle ;
\draw   (330,281) .. controls (330,272.16) and (337.16,265) .. (346,265) -- (394,265) .. controls (402.84,265) and (410,272.16) .. (410,281) -- (410,404) .. controls (410,412.84) and (402.84,420) .. (394,420) -- (346,420) .. controls (337.16,420) and (330,412.84) .. (330,404) -- cycle ;
\draw  [color={rgb, 255:red, 74; green, 144; blue, 226 }  ,draw opacity=1 ] (215.23,161) .. controls (215.23,157.69) and (217.92,155) .. (221.23,155) -- (384.37,155) .. controls (387.68,155) and (390.37,157.69) .. (390.37,161) -- (390.37,179) .. controls (390.37,182.31) and (387.68,185) .. (384.37,185) -- (221.23,185) .. controls (217.92,185) and (215.23,182.31) .. (215.23,179) -- cycle ;
\draw [color={rgb, 255:red, 74; green, 144; blue, 226 }  ,draw opacity=1 ]   (185,265) -- (185,248) ;
\draw [shift={(185,245)}, rotate = 450] [fill={rgb, 255:red, 74; green, 144; blue, 226 }  ,fill opacity=1 ][line width=0.08]  [draw opacity=0] (8.93,-4.29) -- (0,0) -- (8.93,4.29) -- cycle    ;
\draw [color={rgb, 255:red, 126; green, 211; blue, 33 }  ,draw opacity=1 ]   (220,265) -- (220,248) ;
\draw [shift={(220,245)}, rotate = 450] [fill={rgb, 255:red, 126; green, 211; blue, 33 }  ,fill opacity=1 ][line width=0.08]  [draw opacity=0] (8.93,-4.29) -- (0,0) -- (8.93,4.29) -- cycle    ;
\draw [color={rgb, 255:red, 208; green, 2; blue, 27 }  ,draw opacity=1 ]   (255,265) -- (255,248) ;
\draw [shift={(255,245)}, rotate = 450] [fill={rgb, 255:red, 208; green, 2; blue, 27 }  ,fill opacity=1 ][line width=0.08]  [draw opacity=0] (8.93,-4.29) -- (0,0) -- (8.93,4.29) -- cycle    ;
\draw [color={rgb, 255:red, 248; green, 231; blue, 28 }  ,draw opacity=1 ]   (290,265) -- (290,248) ;
\draw [shift={(290,245)}, rotate = 450] [fill={rgb, 255:red, 248; green, 231; blue, 28 }  ,fill opacity=1 ][line width=0.08]  [draw opacity=0] (8.93,-4.29) -- (0,0) -- (8.93,4.29) -- cycle    ;
\draw    (370,265) .. controls (369.65,240.17) and (356.1,230.21) .. (337.91,229.96) ;
\draw [shift={(335,230)}, rotate = 357.66999999999996] [fill={rgb, 255:red, 0; green, 0; blue, 0 }  ][line width=0.08]  [draw opacity=0] (8.93,-4.29) -- (0,0) -- (8.93,4.29) -- cycle    ;
\draw  [color={rgb, 255:red, 126; green, 211; blue, 33 }  ,draw opacity=1 ] (215.23,126.38) .. controls (215.23,123.12) and (217.87,120.48) .. (221.13,120.48) -- (384.46,120.48) .. controls (387.73,120.48) and (390.37,123.12) .. (390.37,126.38) -- (390.37,144.1) .. controls (390.37,147.36) and (387.73,150) .. (384.46,150) -- (221.13,150) .. controls (217.87,150) and (215.23,147.36) .. (215.23,144.1) -- cycle ;
\draw  [color={rgb, 255:red, 208; green, 2; blue, 27 }  ,draw opacity=1 ] (215.23,92.14) .. controls (215.23,88.99) and (217.79,86.43) .. (220.94,86.43) -- (384.65,86.43) .. controls (387.81,86.43) and (390.37,88.99) .. (390.37,92.14) -- (390.37,109.29) .. controls (390.37,112.44) and (387.81,115) .. (384.65,115) -- (220.94,115) .. controls (217.79,115) and (215.23,112.44) .. (215.23,109.29) -- cycle ;
\draw  [color={rgb, 255:red, 248; green, 231; blue, 28 }  ,draw opacity=1 ] (215.23,57.14) .. controls (215.23,53.99) and (217.79,51.43) .. (220.94,51.43) -- (384.65,51.43) .. controls (387.81,51.43) and (390.37,53.99) .. (390.37,57.14) -- (390.37,74.29) .. controls (390.37,77.44) and (387.81,80) .. (384.65,80) -- (220.94,80) .. controls (217.79,80) and (215.23,77.44) .. (215.23,74.29) -- cycle ;
\draw [color={rgb, 255:red, 248; green, 231; blue, 28 }  ,draw opacity=1 ]   (155,65) -- (212,65) ;
\draw [shift={(215,65)}, rotate = 180] [fill={rgb, 255:red, 248; green, 231; blue, 28 }  ,fill opacity=1 ][line width=0.08]  [draw opacity=0] (8.93,-4.29) -- (0,0) -- (8.93,4.29) -- cycle    ;
\draw [color={rgb, 255:red, 248; green, 231; blue, 28 }  ,draw opacity=1 ]   (155,65) -- (155,215) ;

\draw  [dash pattern={on 4.5pt off 4.5pt}] (145,316) .. controls (145,312.69) and (147.69,310) .. (151,310) -- (309,310) .. controls (312.31,310) and (315,312.69) .. (315,316) -- (315,334) .. controls (315,337.31) and (312.31,340) .. (309,340) -- (151,340) .. controls (147.69,340) and (145,337.31) .. (145,334) -- cycle ;
\draw  [dash pattern={on 4.5pt off 4.5pt}] (145,351) .. controls (145,347.69) and (147.69,345) .. (151,345) -- (309,345) .. controls (312.31,345) and (315,347.69) .. (315,351) -- (315,369) .. controls (315,372.31) and (312.31,375) .. (309,375) -- (151,375) .. controls (147.69,375) and (145,372.31) .. (145,369) -- cycle ;
\draw  [dash pattern={on 4.5pt off 4.5pt}] (145,386) .. controls (145,382.69) and (147.69,380) .. (151,380) -- (309,380) .. controls (312.31,380) and (315,382.69) .. (315,386) -- (315,404) .. controls (315,407.31) and (312.31,410) .. (309,410) -- (151,410) .. controls (147.69,410) and (145,407.31) .. (145,404) -- cycle ;
\draw [color={rgb, 255:red, 208; green, 2; blue, 27 }  ,draw opacity=1 ]   (165,100) -- (212,100) ;
\draw [shift={(215,100)}, rotate = 180] [fill={rgb, 255:red, 208; green, 2; blue, 27 }  ,fill opacity=1 ][line width=0.08]  [draw opacity=0] (8.93,-4.29) -- (0,0) -- (8.93,4.29) -- cycle    ;
\draw [color={rgb, 255:red, 208; green, 2; blue, 27 }  ,draw opacity=1 ]   (165,100) -- (165,215) ;

\draw [color={rgb, 255:red, 126; green, 211; blue, 33 }  ,draw opacity=1 ]   (175,135) -- (212,135) ;
\draw [shift={(215,135)}, rotate = 180] [fill={rgb, 255:red, 126; green, 211; blue, 33 }  ,fill opacity=1 ][line width=0.08]  [draw opacity=0] (8.93,-4.29) -- (0,0) -- (8.93,4.29) -- cycle    ;
\draw [color={rgb, 255:red, 126; green, 211; blue, 33 }  ,draw opacity=1 ]   (175,135) -- (175,215) ;

\draw [color={rgb, 255:red, 74; green, 144; blue, 226 }  ,draw opacity=1 ]   (185,170) -- (212,170) ;
\draw [shift={(215,170)}, rotate = 180] [fill={rgb, 255:red, 74; green, 144; blue, 226 }  ,fill opacity=1 ][line width=0.08]  [draw opacity=0] (8.93,-4.29) -- (0,0) -- (8.93,4.29) -- cycle    ;
\draw [color={rgb, 255:red, 74; green, 144; blue, 226 }  ,draw opacity=1 ]   (185,170) -- (185,215) ;

\draw [color={rgb, 255:red, 126; green, 211; blue, 33 }  ,draw opacity=1 ]   (430,325) -- (398,325) ;
\draw [shift={(395,325)}, rotate = 360] [fill={rgb, 255:red, 126; green, 211; blue, 33 }  ,fill opacity=1 ][line width=0.08]  [draw opacity=0] (8.93,-4.29) -- (0,0) -- (8.93,4.29) -- cycle    ;
\draw [color={rgb, 255:red, 126; green, 211; blue, 33 }  ,draw opacity=1 ][fill={rgb, 255:red, 74; green, 144; blue, 226 }  ,fill opacity=1 ]   (430,135) -- (430,325) ;
\draw [color={rgb, 255:red, 126; green, 211; blue, 33 }  ,draw opacity=1 ][fill={rgb, 255:red, 74; green, 144; blue, 226 }  ,fill opacity=1 ]   (430,135) -- (390,135) ;
\draw [color={rgb, 255:red, 208; green, 2; blue, 27 }  ,draw opacity=1 ]   (440,360) -- (398,360) ;
\draw [shift={(395,360)}, rotate = 360] [fill={rgb, 255:red, 208; green, 2; blue, 27 }  ,fill opacity=1 ][line width=0.08]  [draw opacity=0] (8.93,-4.29) -- (0,0) -- (8.93,4.29) -- cycle    ;
\draw [color={rgb, 255:red, 208; green, 2; blue, 27 }  ,draw opacity=1 ][fill={rgb, 255:red, 74; green, 144; blue, 226 }  ,fill opacity=1 ]   (440,100) -- (440,360) ;
\draw [color={rgb, 255:red, 208; green, 2; blue, 27 }  ,draw opacity=1 ][fill={rgb, 255:red, 74; green, 144; blue, 226 }  ,fill opacity=1 ]   (440,100) -- (390,100) ;
\draw [color={rgb, 255:red, 248; green, 231; blue, 28 }  ,draw opacity=1 ]   (450,395) -- (398,395) ;
\draw [shift={(395,395)}, rotate = 360] [fill={rgb, 255:red, 248; green, 231; blue, 28 }  ,fill opacity=1 ][line width=0.08]  [draw opacity=0] (8.93,-4.29) -- (0,0) -- (8.93,4.29) -- cycle    ;
\draw [color={rgb, 255:red, 248; green, 231; blue, 28 }  ,draw opacity=1 ][fill={rgb, 255:red, 74; green, 144; blue, 226 }  ,fill opacity=1 ]   (450,65) -- (450,395) ;
\draw [color={rgb, 255:red, 248; green, 231; blue, 28 }  ,draw opacity=1 ][fill={rgb, 255:red, 74; green, 144; blue, 226 }  ,fill opacity=1 ]   (450,65) -- (390,65) ;

\draw (146,281) node [anchor=north west][inner sep=0.75pt]   [align=left] {$\bfh_1^{(t)}$};
\draw (359,282) node [anchor=north west][inner sep=0.75pt]   [align=left] {$\bfF^{(t)}$};
\draw (359,317) node [anchor=north west][inner sep=0.75pt]   [align=left] {$\bfF^{(t)}$};
\draw (359,352) node [anchor=north west][inner sep=0.75pt]   [align=left] {$\bfF^{(t)}$};
\draw (359,387) node [anchor=north west][inner sep=0.75pt]   [align=left] {$\bfF^{(t)}$};
\draw (176,416) node [anchor=north west][inner sep=0.75pt]   [align=left] {$\bfv_1^{(t)}$};
\draw (217,222) node [anchor=north west][inner sep=0.75pt]  [font=\large] [align=left] {Filter};
\draw (211,416) node [anchor=north west][inner sep=0.75pt]   [align=left] {$\bfv_2^{(t)}$};
\draw (246,416) node [anchor=north west][inner sep=0.75pt]   [align=left] {$\bfv_3^{(t)}$};
\draw (281,416) node [anchor=north west][inner sep=0.75pt]   [align=left] {$\bfv_4^{(t)}$};
\draw (217,161) node [anchor=north west][inner sep=0.75pt]   [align=left] {$\bfv_1^{(t-1)}\shortrightarrow  \bfv_1^{(t-2) }\shortrightarrow \ldots\shortrightarrow  \bfv_1^{(1)}$ };
\draw (217,126) node [anchor=north west][inner sep=0.75pt]   [align=left] {$\bfv_2^{(t-1)}\shortrightarrow  \bfv_2^{(t-2) }\shortrightarrow \ldots\shortrightarrow  \bfv_2^{(1)}$ };
\draw (217,91.5) node [anchor=north west][inner sep=0.75pt]   [align=left] {$\bfv_3^{(t-1)}\shortrightarrow  \bfv_3^{(t-2) }\shortrightarrow \ldots\shortrightarrow  \bfv_3^{(1)}$ };
\draw (217,56.5) node [anchor=north west][inner sep=0.75pt]   [align=left] {$\bfv_4^{(t-1)}\shortrightarrow  \bfv_4^{(t-2) }\shortrightarrow \ldots\shortrightarrow  \bfv_4^{(1)}$ };
\draw (146,315.72) node [anchor=north west][inner sep=0.75pt]   [align=left] {$\bfh_2^{(t)}$};
\draw (146,350.72) node [anchor=north west][inner sep=0.75pt]   [align=left] {$\bfh_3^{(t)}$};
\draw (146,385.72) node [anchor=north west][inner sep=0.75pt]   [align=left] {$\bfh_4^{(t)}$};
\draw (186,157) node [anchor=north west][inner sep=0.75pt]  [font=\scriptsize] [align=left] {Link};
\draw (271,27) node [anchor=north west][inner sep=0.75pt]  [font=\large] [align=left] {Shards};
\draw (334,432) node [anchor=north west][inner sep=0.75pt]   [align=left] {Verification\\ \ \ Results};
\draw (207,447) node [anchor=north west][inner sep=0.75pt]   [align=left] {Block~$\bfB^{(t)}$};
\draw (177.93,283) node [anchor=north west][inner sep=0.75pt]  [font=\scriptsize] [align=left] {$b_{11}$};
\draw (212.93,283) node [anchor=north west][inner sep=0.75pt]  [font=\scriptsize] [align=left] {$b_{12}$};
\draw (247,283) node [anchor=north west][inner sep=0.75pt]  [font=\scriptsize] [align=left] {$b_{13}$};
\draw (282,283) node [anchor=north west][inner sep=0.75pt]  [font=\scriptsize] [align=left] {$b_{14}$};
\draw (177.93,318) node [anchor=north west][inner sep=0.75pt]  [font=\scriptsize] [align=left] {$b_{21}$};
\draw (212.93,318) node [anchor=north west][inner sep=0.75pt]  [font=\scriptsize] [align=left] {$b_{22}$};
\draw (247,318) node [anchor=north west][inner sep=0.75pt]  [font=\scriptsize] [align=left] {$b_{23}$};
\draw (282,318) node [anchor=north west][inner sep=0.75pt]  [font=\scriptsize] [align=left] {$b_{24}$};
\draw (177.93,353) node [anchor=north west][inner sep=0.75pt]  [font=\scriptsize] [align=left] {$b_{31}$};
\draw (212.93,353) node [anchor=north west][inner sep=0.75pt]  [font=\scriptsize] [align=left] {$b_{32}$};
\draw (247,353) node [anchor=north west][inner sep=0.75pt]  [font=\scriptsize] [align=left] {$b_{33}$};
\draw (282,353) node [anchor=north west][inner sep=0.75pt]  [font=\scriptsize] [align=left] {$b_{34}$};
\draw (177.93,388) node [anchor=north west][inner sep=0.75pt]  [font=\scriptsize] [align=left] {$b_{41}$};
\draw (212.93,388) node [anchor=north west][inner sep=0.75pt]  [font=\scriptsize] [align=left] {$b_{42}$};
\draw (247,388) node [anchor=north west][inner sep=0.75pt]  [font=\scriptsize] [align=left] {$b_{43}$};
\draw (282,388) node [anchor=north west][inner sep=0.75pt]  [font=\scriptsize] [align=left] {$b_{44}$};
\draw (390,150) node [anchor=north west][inner sep=0.75pt]  [font=\small] [align=left] {$\bfV_1^{(t)}$};
\draw (390,115) node [anchor=north west][inner sep=0.75pt]  [font=\small] [align=left] {$\bfV_2^{(t)}$};
\draw (390,81) node [anchor=north west][inner sep=0.75pt]  [font=\small] [align=left] {$\bfV_3^{(t)}$};
\draw (390,45) node [anchor=north west][inner sep=0.75pt]  [font=\small] [align=left] {$\bfV_4^{(t)}$};
\draw (183,122) node [anchor=north west][inner sep=0.75pt]  [font=\scriptsize] [align=left] {Link};
\draw (183,87) node [anchor=north west][inner sep=0.75pt]  [font=\scriptsize] [align=left] {Link};
\draw (183,52) node [anchor=north west][inner sep=0.75pt]  [font=\scriptsize] [align=left] {Link};

\end{tikzpicture}

\caption{Illustration of \emph{2-Dimensional Sharding} in a blockchain system with 4 shards. The block~$B^{(t)}$ is horizontally sliced into \emph{outgoing strips}~${\bfh_1^{(t)},\ldots,\bfh_4^{(t)}}$ and vertically sliced into \emph{incoming strips}~$\bfv_1^{(t)},\ldots,\bfv_4^{(t)}$. The outgoing strip~$\bfh_k^{(t)}$ is verified against the corresponding shard~$\bfV_k^{(t)}$ using the verification function~$\bfF^{(t)}$, for~$k=1,2,3,4$. Together, the verification results reveal the validity of every transaction, and help to filter out the invalid transactions in the incoming strips, which are finally linked to the corresponding shards.}
    \label{fig:Fig.1}
\end{figure}
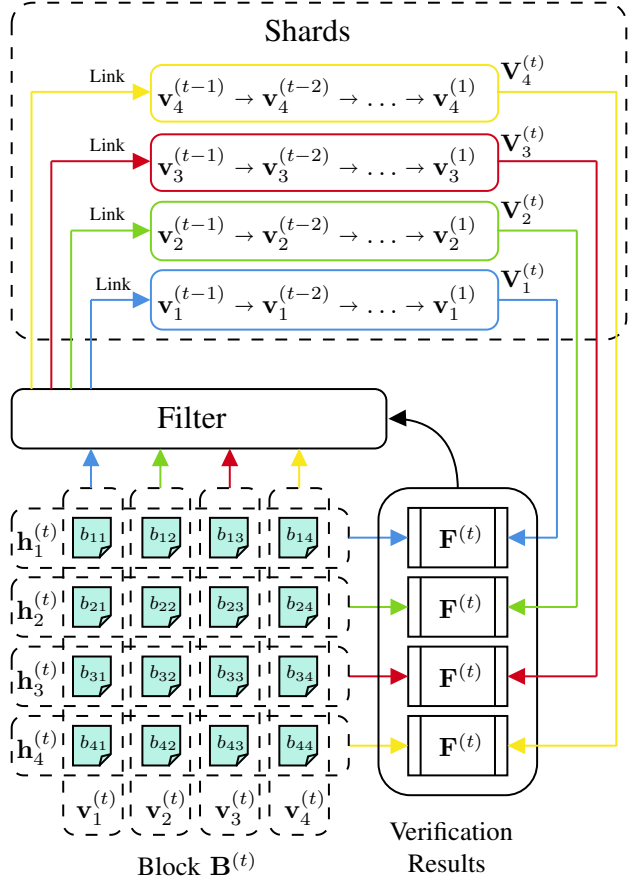

The system includes~$N$ nodes and~$K$ client communities of equal size. 
The nodes are responsible for collecting, verifying and storing transactions; clients issue transactions and transfer coins between each other. 
Note that clients are affiliated with communities, whereas nodes are not.
Transactions are proposed by clients and verified by nodes periodically during time intervals, called \emph{epochs}, denoted by a discrete time unit~$t$. 

We formulate the block containing all transactions in epoch~$t$ as a matrix,
\begin{equation}\label{eq:block}
\bfB^{(t)}=
\begin{bmatrix}
b_{1,1} & b_{1,2} & \ldots & b_{1,K} \\
\vdots &\vdots&\ddots&\vdots\\
b_{K,1} & b_{K,2} & \ldots & b_{K,K} \\
\end{bmatrix},
\end{equation}
where every~$b_{k,r} \in \bbF_q^{Q\times R}$ is a \emph{tiny block}, formed as a concatenation of~$Q$ transactions \new{with senders in community~$k$ and receivers in community~$r$}; each transaction~$\bfx\in\bbF_q^R$ is a vector of length~$R$ over some finite field~$\bbF_q$.

We partition the block~$\bfB^{(t)}$ into \emph{outgoing strips} and \emph{incoming strips}, as shown in Fig.~\ref{fig:Fig.1}.
An outgoing strip
$${\bfh_k^{(t)}=(b_{k,1},\ldots,b_{k,K})\in (\bbF_q^{Q\times R})^K}$$ 
is a vector containing transactions with senders in community~$k$. 
Similarly, an incoming strip
$${\bfv_k^{(t)}=(b_{1,k},\ldots,b_{K,k})\in (\bbF_q^{Q\times R})^K}$$ stands for a collection of all transactions in epoch~$t$ \new{with receivers in community~$k$.}
Equivalently, one can view an outgoing strip~$\bfh^{(t)}_k$ as the~$k$-th row of matrix~$\bfB^{(t)}$, and incoming strip~$\bfv^{(t)}_k$ as the transpose of the~$k$-th column of matrix~$\bfB^{(t)}$, i.e.,
\begin{equation*}
\bfB^{(t)}=\begin{bmatrix}(\bfv^{(t)}_{1})^\intercal,(\bfv^{(t)}_{2})^\intercal,\ldots,(\bfv^{(t)}_{K})^\intercal)\end{bmatrix}
=\begin{bmatrix} (\bfh^{(t)}_{1})^\intercal, & (\bfh^{(t)}_{2})^\intercal, & \ldots ,& (\bfh^{(t)}_{K})^\intercal\end{bmatrix}^\intercal.\\
\end{equation*}

\new{Formally, we define a~\emph{shard}~$\bfV_k^{(t)}=\begin{pmatrix}\bfv_k^{(1)},\ldots,\bfv_k^{(t)}
\end{pmatrix}$ as a concatenation of \emph{incoming} strips associated with community~$k$ from epoch~$1$ to epoch~$t$, which contains~$M(t)$ transactions. 
Note that our definition of a shard is slightly different from the existing literature\footnote{Sharding in blockchain broadly refers to the practice of partitioning nodes among different committees (in a possibly random fashion), each individually handles a portion of verification and storage~\cite{SokSharding}. On the contrary, we do not assign individual node to any specific committee, but partition clients among communities. Meanwhile, we partition transactions based on the community of the receivers, and each partition is called a shard.}. 
This definition provides a inherent support for cross-shard transactions, which will be elaborated in sequel.
}

\begin{remark}
A coded outgoing strip~$\bfh_k^{(t)}$ contains transactions redeeming UTXOs from shard~$\bfV_k^{(t)}$. 
In later sections, we present a polynomial function that verifies~$\bfh_k^{(t)}$ against~$\bfV_k^{(t)}$. 
The results are used to filter out invalid transactions in the incoming strip~$\bfv_k^{(t)}$ before they are appended to the shard. 
As a result, our setting does not differentiate intra- and cross-shard transactions, alleviating the need for sophisticated cross-shard communication mechanisms.
\end{remark}

\subsection{Polynomial Verification Function}\label{subsection:pvf}
In our setting, a new transaction is of the form~$\bfx_{new}=(\bfu_{new},\bfp_{new},\bfa_{new},\bfs_{new})$, where:
\begin{enumerate}
\item~$\bfu_{new}\in\bbF_{q}^{T^{(t-1)}\times 2}$ is a \emph{lookup matrix} used to index the previous transaction, where~$T^{(t)}=\log_2M(t)$.
\item~$\bfp_{new}\in\bbF_q^B$ is the sender's public key, containing all coefficients of an MQ system.
\item~$\bfa_{new}\in\bbF_q^C$ is the receiver's address, i.e., the hash value of the receiver's public key.
\item~$\bfs_{new}\in \bbF_q^D$ is the senders signature on~$\bfx_{new}'=(\bfu_{new},\bfp_{new},\bfa_{new})$
\end{enumerate}

Verifying~$\bfx_{new}$ includes three crucial parts:

\begin{itemize}
    \item \textbf{Transaction Fetching}: To fetch the corresponding old transaction~$\bfx_{old}=(\bfu_{old},\bfp_{old},\bfa_{old},\bfs_{old})$ from which~$\bfx_{new}$ redeems the~UTXO.
    \item \textbf{Address Checking}: To check whether the hash value of~$\bfp_{new}$ matches~$\bfa_{old}$.
    \item \textbf{Signature Verification}: To verify that~$\bfs_{new}$ is a valid signature on the hash value of~$\bfx_{new}'$ by using the public key~$\bfp_{new}$.
\end{itemize}

In detail, the above parts are executed as follows.

\subsubsection{Transaction Fetching}
The lookup matrix~$\bfu_{new}$ has exactly one~1-entry and one~0-entry in each row. Hence, every transaction in~$\bfV^{(t)}$ can be uniquely indexed by a lookup matrix. The verifier views shard~$\bfV^{(t)}$ as~$T^{(t-1)}$-dimensional tensor in~$(\bbF_{q^R})^{2\times 2\times ...\times 2}$, and therefore every transaction can be conveniently expressed as a tensor entry~${\bfV^{(t)}_{i_1,\ldots,i_{T^{(t-1)}}}\in \bbF_{q^R}}$.

To fetch a transaction, one computes a multilinear polynomial,
\begin{equation*}
    \begin{split}
{\Var{fetch}}^{(t)}(\bfu,\bfV^{(t)})= \sum_{(i_1,\ldots,i_{T^{(t-1)}})\in \{1,2\}^{T^{(t-1)}}} \left(\prod_{j=1}^{T^{(t-1)}} \bfu_{j,i_j}\right) \bfV^{(t)}_{i_1,\ldots,i_{T^{(t-1)}}}
    \end{split}
\end{equation*}
which takes a shard~$\bfV^{(t)}$ and a lookup table~$\bfu$ as inputs and yields the transaction~$\bfx_{\bfu}\in\bbF_{q^R}$ indexed by~$\bfu$. The degree of~${\Var{fetch}}^{(t)}$ is~$T^{(t-1)}+1$. Note that the subscript~$k$ is omitted in~${\Var{fetch}}^{(t)}$ since it can be applied to any shard. 

\begin{example}
In a shard~$\bfV$ that contains~$8$ transactions, to fetch one of them, one would compute
\begin{equation*}
    \begin{split}
        \Var{fetch}(\bfu,\bfV)&=
                \bfu_{1,1}\bfu_{2,1}\bfu_{3,1}~\bfV_{1,1,1}+
                 \bfu_{1,1}\bfu_{2,1}\bfu_{3,2}~\bfV_{1,1,2}+
                 \bfu_{1,1}\bfu_{2,2}\bfu_{3,1}~\bfV_{1,2,1}+
                 \bfu_{1,1}\bfu_{2,2}\bfu_{3,2}~\bfV_{1,2,2}\\
                &\phantom{=}+\bfu_{1,2}\bfu_{2,1}\bfu_{3,1}~\bfV_{2,1,1}+
                \bfu_{1,2}\bfu_{2,1}\bfu_{3,2}~\bfV_{2,1,2}+
                \bfu_{1,2}\bfu_{2,2}\bfu_{3,1}~\bfV_{2,2,1}+
                \bfu_{1,2}\bfu_{2,2}\bfu_{3,2}~\bfV_{2,2,2}.
    \end{split}    
\end{equation*}
Since the lookup matrix contains only one 1-entry and one 0-entry, only the entry indexed by~$\bfu$ has coefficient~$1$, while the rest have coefficients~$0$.
\end{example}
\subsubsection{Address Checking}

Based on Section \ref{subsection:cryptoPrimitives}, we assume a multivariate polynomial~${\Var{hash1}}:\bbF_q^B\shortrightarrow \bbF_q^C$ of a constant degree to serve as our first collision resistant hash function. 
Having obtained~$\bfx_{old}={\Var{fetch}}^{(t)}(\bfu_{new},\bfV^{(t)}_k)$, the verifier then checks whether~${\Var{hash1}}(\bfp_{new})=\bfa_{old}$ holds, which is expressed as a polynomial,
$${\Var{checkAddr}}(\bfp,\bfa)={\Var{hash1}}(\bfp)-\bfa.$$

Note that~$\bfp_{new}$ is accepted when~${{\Var{checkAddr}}(\bfp_{new},\bfa_{old})\in\bbF_q^C}$ is the all-zero vector. 

\subsubsection{Signature Verification}
The verifier needs to check the validity of the signature~$\bfs_{new}$. She first computes a hash digest~$\bfw={\Var{hash2}}(\bfu_{new}, \bfp_{new},\bfa_{new})=(w_1,\ldots,w_E)\in \bbF_q^E$, where~${\Var{hash2}}: \bbF_q^{A+B+C}\shortrightarrow \bbF_q^E$ is our second collision resistant hash function of a constant degree. Later, the verifier checks whether~$MQ(\bfp_{new},\bfs_{new})=\bfw$ holds, where,
$$
MQ(\bfp,\bfs)= \sum_{0<i\leq j<D}\bfa_{(i,j)}s_is_j+\sum_{0<i<D}\bfb_is_i+\bfc,
$$
and~$\bfa, \bfb, \bfc \in \bbF_q^E$ are vectors stored in~$\bfp_{new}$, serving as coefficients of the MQ problem. Equivalently, the verification of a signature~$\bfs$ in a transaction~$\bfx=(\bfu,\bfp,\bfa,\bfs)$ can be expressed as a polynomial,
$$
{\Var{checkSig}}(\bfx)= MQ(\bfp,\bfs)-{\Var{hash2}}(\bfu,\bfp,\bfa).
$$
Note that~$\bfs_{new}$ is accepted only when~${\Var{checkSig}}(\bfx_{new})=0$.

The above three parts focus on the verification of an individual transaction. We further employ them to verify the entire strip as follows.

\subsubsection{Verification of a strip}
Let~$\eta\in \bbF_{q}^{C+E}$ be the concatenation of~${\Var{checkAddr}}(\bfp_{new},\bfa_{old})$ and ${\Var{checkSig}}(\bfx_{new})$; a transaction~$\bfx$ is accepted if and only if~$\eta=f^{(t)}(\bfx,\bfV^{(t)})=0$. 
This information is further used to fliter out invalid trancations in the incoming strip.
Since the UTXOs redeemed by transactions in~$\bfh^{(t)}_{k}=(\bfx_1,\ldots,\bfx_{QK})$ all reside in~$\bfV^{(t)}_{k}$, we define a multivariate polynomial,
\begin{align*}
	F^{(t)}(\bfh^{(t)},\bfV^{(t-1)})=
	(f^{(t)}(\bfx_1,\bfV^{(t-1)})^\intercal,\ldots,f^{(t)}(\bfx_{QK},\bfV^{(t-1)})^\intercal),	
\end{align*}
~\new{of degree~$d$}, which yields an \emph{outgoing result strip}
\begin{equation}\label{eq:outgoingResultStrip}
\bfe^{(t)}_k=(r_{k,1},\ldots,r_{k,K}) \in (\bbF_q^{Q\times (C+E)})^K,
\end{equation}
defined as the~$k$-th row of the \emph{result matrix}
\begin{equation}\label{eq:result block}
\bfR^{(t)}=
\begin{bmatrix}
r_{1,1} & r_{1,2} & \ldots & r_{1,K} \\
\vdots &\vdots&\ddots&\vdots\\
r_{K,1} & r_{K,2} & \ldots & r_{K,K} \\
\end{bmatrix}.
\end{equation}

Each~\emph{tiny result block}~$r_{k,k'}$ contains~$Q$ entries of length $C+E$; one per every transaction in the tiny block~$b_{k,k'}$. The~$j$-th entry in~$r_{k,k'}$ is the result of computing~$f^{(t)}$ on~$\bfV^{(t)}_k$ and the~$j$-th transaction in~$b_{k,k'}$. Hence, the outgoing result strip~$\bfe^{(t)}_k$ reveals the validity of every transaction in the outgoing strip~$\bfh_k^{(t)}$, and the result matrix~$\bfR^{(t)}$ reveals the validity of every transaction in the block~$\bfB^{(t)}$. As shown in Fig.~\ref{fig:Fig.1}, the outgoing result strips are used to filter out the invalid transactions in the incoming strips before they are being appended to the corresponding shards. 

Similarly, the~\emph{incoming result strip}
\begin{equation}\label{eq:incoming result strip}
    \bfs_k^{(t)}=(r_{1,k},\ldots,r_{K,k})\in (\bbF_q^{Q\times (C+E)})^K
\end{equation}
is a transpose of the~$k$-th column of the result block. It reveals the validity of every transaction in the incoming strip~$\bfv_k^{(t)}$. We will employ this notation in Section~\ref{subsection:codedAppending}.

\begin{remark} [The degree of $F^{(t)}$]
\new{
The verification result of a transaction is the concatenation of functions~$\Var{checkAddr}$ and~$\Var{checkSig}$. By definition,~${\Var{checkAddr}}(\bfp,\bfa)={\Var{hash1}}(\bfp)-\bfa$, where~$\bfa$ is the output of function~$\Var{fetch}^{(t)}$.
Besides,~${\Var{checkSig}}(\bfx)= MQ(\bfp,\bfs)-{\Var{hash2}}(\bfu,\bfp,\bfa)$, where the degree of the multivariate function~$MQ(\bfp,\bfs)$ is~$3$.
Together, the degree of polynomial that verifies a transaction is
~$
d=\max(T^{(t-1)}+1,\deg \Var{hash1}, \deg \Var{hash2},3),
$
where ~$T^{(t-1)}+1$ is the degree of~$\Var{fetch}^{(t)}$.

As existing works~\cite{LC,MQAnalysis,MQHash} show the existence of secure polynomial hash functions with degree as low as~$3$, we assume that the degree of both $\Var{hash1}$ and~$\Var{hash2}$ is less than~$T^{(t-1)}$ in realistic blockchain systems (e.g., a blockchain system with~$10^6$ transactions in each shard has a verification function of degree~$d=20$). 
Hence, the polynomial~$F^{(t)}$ has a degree~$d=T^{(t-1)}+1$, which scales \emph{logarithmically} with the number of transactions in a shard.}
\end{remark}

\subsection{Coded Computation}\label{subsection:codedComputation}
Now that the verification of outgoing strips has been formulated as a low degree polynomial, we turn to describe how it is conducted in a coded fashion. In detail, every shard~$k\in[K]$ is assigned a unique scalar~$\omega_k\in\bbF_q$, and every node~$i\in[N]$ is assigned a unique scalar~$\alpha_i\in\bbF_q$.

Setting~$T=0$, the generator matrix in~\eqref{eq:generator matrix} becomes
\begin{equation}\label{eq:GL}
G_\cL=
\begin{bmatrix}
\Phi_1(\alpha_1) & \Phi_1(\alpha_2) & \ldots &\Phi_1(\alpha_N) \\
\vdots &\vdots&\ddots&\vdots\\
\Phi_{K}(\alpha_1) & \Phi_K(\alpha_2) & \ldots & \Phi_{K}(\alpha_N) \\
\end{bmatrix},
\end{equation}
where~$\Phi_k(z)$ is the Lagrange polynomial~\eqref{equation:LagrangePolynomial}.
For node~$i$, the coded outgoing strip and coded incoming strip are linear combinations of outgoing strips and incoming strips, respectively, i.e., 
\begin{align*}
(\coded{\bfh}_i^{(t)})^\intercal & =((\bfh_1^{(t)})^\intercal,\ldots,(\bfh_K^{(t)})^\intercal)\cdot (G_\cL)_i=(\bfB^{(t)})^\intercal\cdot (G_\cL)_i,\\
(\coded{\bfv}_i^{(t)})^\intercal & =((\bfv_1^{(t)})^\intercal,\ldots,(\bfv_K^{(t)})^\intercal)\cdot (G_\cL)_i\:=\bfB^{(t)}\cdot (G_\cL)_i,
\end{align*}
where~$(G_\cL)_i$ is the~$i$-th column of~$G_\cL$. Equivalently,~$\coded{\bfh}^{(t)}_{i}$ and~$\coded{\bfv}^{(t)}_{i}$ are evaluations of Lagrange polynomials~$\psi^{(t)}(z)$ and~$\phi^{(t)}(z)$ at~$\alpha_i$, respectively, where
\begin{equation*}
\psi^{(t)}(z) =\sum_{k=1}^K\bfh^{(t)}_{k}\prod_{j\neq k} \frac{z-\omega_j}{\omega_k - \omega_j}~\text{and}~
\phi^{(t)}(z) =\sum_{k=1}^K\bfv^{(t)}_{k}\prod_{j\neq k} \frac{z-\omega_j}{\omega_k - \omega_j}.
\end{equation*}

Every node~$i$ stores a coded shard~$\coded{\bfV}^{(t)}_{i}$, i.e., a node-specific linear combination of all shards,
$$\coded{\bfV}^{(t)}_{i}=\sum_{k=1}^K G_{k,i}\bfV^{(t)}_k=(\phi^{(1)}(\alpha_i),\ldots,\phi^{(t)}(\alpha_i)).$$

In epoch~$t$, every node~$i$ receives the coded strips~$\coded{\bfh}^{(t)}_{i}$ and~$\coded{\bfv}^{(t)}_{i}$; protocols for secure encoding and delivery of coded strips are given in Section~\ref{section:communication}. 
Node~$i$ computes the polynomial verification function~$F^{(t)}$ on~$\coded{\bfh}^{(t)}_{i}$ and the locally stored~$\coded{\bfV}_i^{t-1}$, and obtains a~\emph{coded outgoing result strip}
\begin{equation*}
\coded{\bfe}^{(t)}_{i}=F^{(t)}(\coded{\bfh}^{(t)}_{i}, \coded{\bfV}_i^{(t-1)})=F^{(t)}(\psi^{(t)}(\alpha_i),(\phi^{(1)}(\alpha_i),\ldots,\phi^{(t)}(\alpha_i)).
\end{equation*}

Formally, the coded outgoing result strip~$\coded{\bfe}^{(t)}_{i}$, as well as the (uncoded) outgoing result strip~$\bfe_k^{(t)}$ defined in Equation~\eqref{eq:outgoingResultStrip}, is an evaluation of a polynomial~$\bfF$, i.e.,
\begin{equation}\label{equation:FeCodedUncoded}
    \bfe_k^{(t)}=\bfF^{(t)}(\omega_k)~\text{and}~\coded{\bfe}^{(t)}_{i}=\bfF^{(t)}(\alpha_i),
\end{equation}
\begin{equation}\label{eq:verificationFunction}
    \text{where}~\bfF^{(t)}(z)=F^{(t)}(\psi^{(t)}(z),(\phi^{(1)}(z),\ldots,\phi^{(t)}(z)).
\end{equation}
\new{Since the degree of both~$\psi^{(t)}$ and~$\phi^{(t)}$ is~$K-1$, it follows that the degree of~$\bfF^{(t)}$ is~$(K-1)d$.}

Similar to the outgoing result strip defined in~\eqref{eq:outgoingResultStrip}, the coded outgoing result strip
\begin{equation}\label{eq:codedOutgoingResultStrip}
\coded{\bfe}^{(t)}_{i}=(\coded{\bfe}^{(t)}_{i,1},\ldots,\coded{\bfe}^{(t)}_{i,K}) \in (\bbF_q^{Q\times (C+E)})^K
\end{equation}
is a length-$K$ vector, in which the~$k$-th element~$\coded{\bfe}^{(t)}_{i,k}$ contains~$Q$ entries and equals to the verification result of the~$k$-th coded tiny block in the coded outgoing strip~$\coded{\bfh}^{(t)}_{i}$. 
Note that unlike the coded incoming strip or the coded outgoing strip, the coded outgoing result strip~$\coded{\bfe}^{(t)}_{i}$ is \textit{not} a linear combination of outgoing result strips~$\bfe^{(t)}_1,\ldots,\bfe^{(t)}_K$ specified by~$G_\cL$.
Instead, both of the coded~$\coded{\bfe}^{(t)}_{i}=\bfF(\alpha_i)$ and uncoded~${\bfe}^{(t)}_{i}=\bfF(\omega_i)$ are evaluation of polynomial~$\bfF(z)$ at different points (see Equation~\eqref{eq:codedResult} and~\eqref{eq:decodeResult} for details).

Nodes further obtain the~\emph{indicator vector}~$g\in\{0,1\}^{QK}$ by exchanging~$\coded{\bfe}^{(t)}_{i}$; the details are given in Section~\ref{subsection:maintainingValidity}. This data is crucial for the next section, namely~\emph{coded appending}, as each of its entries is associated with a coded transaction in every coded incoming strip. 
Specifically, note that every coded transaction is a linear combination of~$K$ transactions; the corresponding entry of the indicator vector~$g$ equals to~$0$ if they are all valid. 
Otherwise, if invalid transactions are included, the entry equals to~$1$.


\subsection{Coded Appending}\label{subsection:codedAppending}

The appending operation of node~$i$ is instructed by the indicator vector~$g$. 
Following the coded verification, each node~$i$ appends the coded \emph{incoming} strip~$\coded{\bfv}_i^{(t)}$ to their coded shard, after setting to zero the parts of it which failed the verification process. That is, node~$i$ zeros out the transactions whose corresponding entry of $g$ equals to~$1$.

\begin{remark}
This process of setting to zero the parts which fail verification has an unexpected implication---it invalidates valid transactions that were linearly combined with invalid ones. We define the \emph{Collateral Invalidation} (CI) rate as the number of transactions that are abandoned due to one invalid transaction, normalized by the total number of transactions processed in one epoch. Polyshard~\cite{Polyshard} has an CI rate of~$\frac{1}{K}$, while our scheme has an CI rate of~$\frac{1}{KQ}$, which is~$Q$ times smaller (i.e., better) than Polyshard. 
\end{remark}

\section{Coded Consensus}\label{section:communication}
In this section, we discuss the consensus aspect of our design. 
Due to its coded nature, we propose three conditions that define~\emph{coded consensus}. 
Later, we show mechanisms that maintain these conditions. 
We use~$f$ to denote the number of Byzantine nodes, and define a~\emph{quorum} as a set of~$N-f$ nodes. 
Our scheme tolerates these~$f$ Byzantine nodes in the partial synchrony
model, given that~$N\ge (K-1)d+3f+1$; \new{a discussion on the nature of this assumption is given in the following section. 
Note that the communication between nodes is point-to-point, and the communication analysis takes into consideration every bit that is transmitted through the system.}

First, our design must guarantee~\emph{consistency}, i.e, at every epoch~$t$, correct nodes must perform coded verification on coded outgoing strips generated from the same block~$\bfB^{(t)}$. 
Formally, we propose the following condition.

\begin{condition}[\textbf{Consistency}]\label{condition:consistency}
Every correct node~$i$ obtains~$\coded{\bfh}_i$, defined as
$\coded{\bfh}_i^{(t)}= (G_\cL)_i^\intercal\cdot \bfB^{(t)}$, where~$(G_\cL)_i$ denotes the~$i$-th column of the generator matrix~$G_\cL$.
\end{condition}

This condition imposes that correct nodes obtain coded outgoing strips that are \emph{consistent} with each other, i.e., correspond to the same block~$\bfB^{(t)}$ defined in Equation~\eqref{eq:block}.
Otherwise, correct verification is impossible, as suggested in~\cite{Discrepancy}.
Moreover, our design must maintain~\emph{homology}.
\begin{condition}[\textbf{Homology}]\label{condition:homology}
Every correct node~$i$ obtains~$\coded{\bfv}_i$, defined as
$\coded{\bfv}_i= (G_\cL)_i^\intercal(\bfB^{(t)})^\intercal$, where both~$\bfB^{(t)}$ and~$G_\cL$ are as in Condition~\ref{condition:consistency}.
\end{condition}

The second condition suggests that every node obtains the coded incoming strip that is~\emph{homologous} to the coded outgoing strip, i.e., generated from the same block~$\bfB^{(t)}$.
Otherwise, we say they are~\emph{nonhomologous}; such nonhomology problem can cause a discrepancy between the verified and the appended, i.e., nodes verify valid transactions, but append invalid ones, nullifying the verification efforts.
Satisfying this condition assures the correct appending of incoming strips. Finally, the blockchain must not store invalid transactions, which gives rise to the last condition.

\begin{condition}[\textbf{Validity}]\label{condition:validity}
Every correct node~$i$ appends the coded incoming strip~$\coded{\bfv}_i^{(t)}$ to its local coded chain after setting the invalid coded transactions to zero, i.e, coded transactions which were not formed exclusively from valid transactions.
\end{condition}
This condition requires every node~$i$ to obtain the indicator vector~$g$ defined in Section~\ref{subsection:codedAppending}.
Together, we say that a protocol provides~\emph{coded consensus} if it simultaneously achieves Condition~\ref{condition:consistency}, Condition~\ref{condition:homology} and Condition~\ref{condition:validity}, i.e., maintains consistency, homology, and validity at the same time. We propose such a protocol, employing a leader to distribute coded strips and provide coded consensus. Our approach adopts HotStuff~\cite{Hotstuff}, a BFT SMR with linear message complexity, and techniques from Information Dispersal for consistency and homology. In addition, we employ coded computation that maintains validity of the system.
Further, the superscript~$(t)$ is omitted for clarity in the rest of this paper.

\subsection{Overview}

\new{In order to maintain the aforementioned three properties, we employ HotStuff to maintain a chain of~\emph{headers}, each corresponds to a block.
HotStuff provides the safety and liveness property of the header chain. 
Together with information dispersal techniques, our scheme maintains the consistency property. 
We provide detailed discussion in Section~\ref{subsection:maintainingConsistency}.
Further, we incorporate extra mechanisms in HotStuff to maintain homology (Section~\ref{subsection:maintainingHomology}) and validity (Section~\ref{subsection:maintainingValidity}).
In Section~\ref{section:discussion}, we show that our scheme indeed provides coded consensus, and inherits the liveness property from HotStuff.
}

\begin{figure}
    \centering

\tikzset{every picture/.style={line width=0.75pt}} 

\begin{tikzpicture}[x=0.75pt,y=0.75pt,yscale=-1,xscale=1]

\draw   (261.4,71) -- (329.4,71) -- (329.4,102.07) -- (261.4,102.07) -- cycle ;
\draw   (381.73,71.67) -- (450.07,71.67) -- (450.07,102.73) -- (381.73,102.73) -- cycle ;
\draw    (381.4,87.27) -- (335.4,87.27) ;
\draw [shift={(333.4,87.27)}, rotate = 360] [color={rgb, 255:red, 0; green, 0; blue, 0 }  ][line width=0.75]    (10.93,-3.29) .. controls (6.95,-1.4) and (3.31,-0.3) .. (0,0) .. controls (3.31,0.3) and (6.95,1.4) .. (10.93,3.29)   ;
\draw   (501.07,71.67) -- (569.4,71.67) -- (569.4,102.73) -- (501.07,102.73) -- cycle ;
\draw    (500.73,87.27) -- (454.73,87.27) ;
\draw [shift={(452.73,87.27)}, rotate = 360] [color={rgb, 255:red, 0; green, 0; blue, 0 }  ][line width=0.75]    (10.93,-3.29) .. controls (6.95,-1.4) and (3.31,-0.3) .. (0,0) .. controls (3.31,0.3) and (6.95,1.4) .. (10.93,3.29)   ;
\draw   (621.07,71) -- (689.4,71) -- (689.4,102.07) -- (621.07,102.07) -- cycle ;
\draw    (620.73,86.6) -- (574.73,86.6) ;
\draw [shift={(572.73,86.6)}, rotate = 360] [color={rgb, 255:red, 0; green, 0; blue, 0 }  ][line width=0.75]    (10.93,-3.29) .. controls (6.95,-1.4) and (3.31,-0.3) .. (0,0) .. controls (3.31,0.3) and (6.95,1.4) .. (10.93,3.29)   ;
\draw   (261.4,131.67) -- (329.4,131.67) -- (329.4,162.73) -- (261.4,162.73) -- cycle ;
\draw   (381.73,132.33) -- (450.07,132.33) -- (450.07,163.4) -- (381.73,163.4) -- cycle ;
\draw   (501.07,132.33) -- (569.4,132.33) -- (569.4,163.4) -- (501.07,163.4) -- cycle ;
\draw   (621.07,131.67) -- (689.4,131.67) -- (689.4,162.73) -- (621.07,162.73) -- cycle ;
\draw    (295.4,131.27) -- (295.4,105.93) ;
\draw [shift={(295.4,103.93)}, rotate = 90] [color={rgb, 255:red, 0; green, 0; blue, 0 }  ][line width=0.75]    (10.93,-3.29) .. controls (6.95,-1.4) and (3.31,-0.3) .. (0,0) .. controls (3.31,0.3) and (6.95,1.4) .. (10.93,3.29)   ;
\draw    (416.07,131.93) -- (416.07,106.6) ;
\draw [shift={(416.07,104.6)}, rotate = 90] [color={rgb, 255:red, 0; green, 0; blue, 0 }  ][line width=0.75]    (10.93,-3.29) .. controls (6.95,-1.4) and (3.31,-0.3) .. (0,0) .. controls (3.31,0.3) and (6.95,1.4) .. (10.93,3.29)   ;
\draw    (536.07,132.6) -- (536.07,107.27) ;
\draw [shift={(536.07,105.27)}, rotate = 90] [color={rgb, 255:red, 0; green, 0; blue, 0 }  ][line width=0.75]    (10.93,-3.29) .. controls (6.95,-1.4) and (3.31,-0.3) .. (0,0) .. controls (3.31,0.3) and (6.95,1.4) .. (10.93,3.29)   ;
\draw    (655.4,131.27) -- (655.4,105.93) ;
\draw [shift={(655.4,103.93)}, rotate = 90] [color={rgb, 255:red, 0; green, 0; blue, 0 }  ][line width=0.75]    (10.93,-3.29) .. controls (6.95,-1.4) and (3.31,-0.3) .. (0,0) .. controls (3.31,0.3) and (6.95,1.4) .. (10.93,3.29)   ;

\draw (269.07,77.67) node [anchor=north west][inner sep=0.75pt]  [font=\large] [align=left] {\header};
\draw (282.4,136.33) node [anchor=north west][inner sep=0.75pt]   [align=left] {$\coded{\bfv}_i^{t-2}$};
\draw (402.4,137.67) node [anchor=north west][inner sep=0.75pt]   [align=left] {$\coded{\bfv}_i^{t-1}$};
\draw (522.4,137.67) node [anchor=north west][inner sep=0.75pt]   [align=left] {$\coded{\bfv}_i^{t}$};
\draw (643.07,137) node [anchor=north west][inner sep=0.75pt]   [align=left] {$\coded{\bfv}_i^{t+1}$};
\draw (194.57,86.67) node  [font=\Large] [align=left] {\begin{minipage}[lt]{97.69pt}\setlength\topsep{0pt}
Header Chain
\end{minipage}};
\draw (389.73,78.33) node [anchor=north west][inner sep=0.75pt]  [font=\large] [align=left] {\header};
\draw (509.07,77.67) node [anchor=north west][inner sep=0.75pt]  [font=\large] [align=left] {\header};
\draw (629.07,77.67) node [anchor=north west][inner sep=0.75pt]  [font=\large] [align=left] {\header};
\draw (200.07,148) node  [font=\Large] [align=left] {\begin{minipage}[lt]{97.92pt}\setlength\topsep{0pt}
Coded Shard
\end{minipage}};

\end{tikzpicture}

\caption{Illustration of the internal storage of node~$i$. We use HotStuff for nodes to reach a consensus on a chain of headers. Meanwhile, each node stores a distinct chain of coded incoming strips. i.e., the coded shard. The consensus on the chain of headers assures that at any height~$t$ (i.e., epoch) of the chain, nodes store coded outgoing strips that are consistent with each other, i.e., generated from the same block~$\bfB$}.
    \label{fig:Fig.2}
\end{figure}
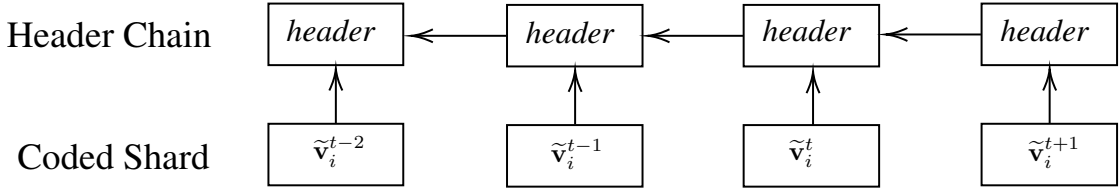

\subsection{Maintaining Consistency (Condition~\ref{condition:consistency})}\label{subsection:maintainingConsistency}

We first address the consistency problem of coded outgoing strips generated from~$\bfB$ (the case for coded incoming strips are similar). 
Our method depends on a data structure called~\emph{checksum}, introduced by AVID-FP~\cite{AVID-FP}. 
A checksum allows nodes to verify that the received coded strip is consistent with ones received by others, i.e., computed from the same block.
It contains a list of~$K$ \emph{fingerprints}. 
Each fingerprint is generated from an (uncoded) outgoing strip (a row of~$\bfB$), using some~$\varepsilon$-fingerprinting function~$\Var{fp}$ defined as follows.

\begin{definition}
\cite[Definition~2.1]{AVID-FP} A function~$\Var{fp}:T\times \bbF_q^\delta \shortrightarrow  \bbF_q^\gamma$ is an~$\varepsilon$-fingerprinting function if
$$
\max_{d,d'\in \bbF_q^\delta,  d\neq d'} \Pr_{r \sim Unif(T)}
[\Var{fp}(r,d)=\Var{fp}(r,d')]\leq \varepsilon.
$$
That is, the probability for two distinct~$d,d'\in\bbF_q^\delta$ to have the same fingerprint is at most~$\varepsilon$, where the key~$r$ is chosen uniformly at random from some input space~$T$. 
\end{definition}

Examples of $\varepsilon$-fingerprinting functions include \emph{division fingerprinting}, which generalizes Rabin's fingerprinting~\cite{Rabin's Fingerprinting} from~$\bbF_2$ to any field~$\bbF_q$.  With coefficients in~$\bbF_q$, the input~$d\in\bbF_q^\delta$ is regarded as a polynomial~$d(x)$ of degree~$\delta$, and~$T$ is a collection of monic irreducible polynomials of degree~$\gamma$. The division fingerprinting function returns the remainder of~$d(x)$ divided by~$p(x)$, i.e.,~$d(x)\text{ mod } p(x)$, where~$p(x)$ is chosen from~$T$ uniformly at random.

We let~$\Var{fp}:T\times \bbF_q^\delta \shortrightarrow  \bbF_q^\gamma$ be an~$\varepsilon$-fingerprinting function where~$\delta=\frac{|\bfB|}{K}$ is the size of a strip. 
As done in AVID-FP~\cite{AVID-FP}, the random selection of~$r$ from~$T$ is simulated by deterministic cryptographic hash functions~\cite{Random Selection}; 
referring to the use of a hash function as a random oracle 
is a common practice in blockchain systems, e.g., in~\cite{Elastico}. 
In addition to the~$K$ fingerprints, a list of hash values~$cc=[\Var{hash}(\coded{\bfh}_1),\ldots,\Var{hash}(\coded{\bfh}_N)]$ is included in the checksum, generated using a cryptographic hash function~$hash:\bbF_q^*\shortrightarrow \bbF_q^\lambda$ (not to be confused with~$\Var{hash1}$ and~$\Var{hash2}$ mentioned earlier). 
The selection of~$r$ is achieved by another cryptographic hash function~$\Var{select}: (\bbF_q^\lambda)^N\shortrightarrow  T$, which takes the list~$cc$ as input and outputs an element in~$T$\footnote{We point out that Verifiable Random Function (VRF)~\cite{VRF} is an alternative implementation of the random oracle, which is communication-efficient as it does not require the checksum to contain~$N$ hash values. 
It has been employed in blockchain designs including~\cite{Algorand} and~\cite{Omniledger}, but for different purposes.}. 
Formally, the function~\FUNC{Checksum} in Algorithm~\ref{alg:utilities} encapsulates the construction of the checksum. 

The leader first generates coded outgoing strips~$(G_\cL)_1^\intercal\cdot\bfB,\ldots,(G_\cL)_N^\intercal\cdot\bfB$ and constructs the checksum~$\Var{cksH}$. 
Similarly, it creates coded incoming strips~$(G_\cL)_1^\intercal\cdot\bfB^\intercal,\ldots,(G_\cL)_N^\intercal\cdot\bfB^\intercal$ and~$\Var{cksV}$.
The leader then sends the checksums to every node~$i$ piggybacked with the coded fragments~$\coded{\bfh}_i=(G_\cL)_i^\intercal\cdot\bfB$ and~$\coded{\bfv}_i=(G_\cL)_i^\intercal\cdot\bfB^\intercal$.
In order to verify that a coded strip \emph{agrees} with the received checksum, i.e.,~$\Var{cksH}$ and~$\coded{\bfh}_i$ are computed from the same~$\bfB$, and~$\Var{cksV}$ and~$\coded{\bfv}_i$ are computed from the same~$\bfB^\intercal$, we require the fingerprinting function to be homomorphic. 

\begin{definition}
\cite[Definition~2.5]{AVID-FP} A fingerprinting function~$\Var{fp}:T\times \bbF_q^\delta \shortrightarrow  \bbF_q^\gamma$ is \textit{homomorphic} if~$\Var{fp}(r,d)+\Var{fp}(r,d') = \Var{fp}(r,d +d')$ and~$b\cdot \Var{fp}(r,d)=\Var{fp}(r,b\cdot d)$ for any~$r\in T$, any~$b\in \bbF_q$, and any~$d,d'\in \bbF_q^\delta$. 
\end{definition}

This property enables the node to verify that the coded fragment satisfies the required linear combination (defined by the generator matrix) with uncoded strips, by having access only to the fingerprints of uncoded strips, and not to the uncoded strips.
As any coded strip is a linear combination of~$K$ uncoded strips, the homomorphism guarantees that its fingerprint must be equal to the same linear combination of~$K$ fingerprints of uncoded strips, i.e.,
\begin{equation*}
    (G_\cL)^\intercal \begin{bmatrix} \Var{fp}(r,{\bfh}_{1})\\\vdots\\\Var{fp}(r,{\bfh}_{K})\end{bmatrix} = 
    \begin{bmatrix} \Var{fp}(r,\coded{\bfh}_{1})\\\vdots\\\Var{fp}(r,\coded{\bfh}_{N})\end{bmatrix},~\mbox{and~}
    (G_\cL)^\intercal \begin{bmatrix} \Var{fp}(r,{\bfv}_{1})\\\vdots\\\Var{fp}(r,{\bfv}_{K})\end{bmatrix} = 
    \begin{bmatrix} \Var{fp}(r,\coded{\bfv}_{1})\\\vdots\\\Var{fp}(r,\coded{\bfv}_{N})\end{bmatrix}.
\end{equation*}
Note that~${\bfh}_{k}$ is the~$k$-th row of the matrix~$\bfB$, and~${\coded{\bfh}}_{i}$ is the~$i$-th row of the matrix~$\coded{\bfB}=(G_\cL)_i^\intercal\bfB$.
Similarly,~${\bfv}_{k}$ is the~$k$-th row of the matrix~$\bfB^\intercal$, and~${\coded{\bfv}}_{i}$ is the~$i$-th row of the matrix~$\coded{\bfB^\intercal}=(G_\cL)_i^\intercal\bfB^\intercal$.
Thus, each node can confirm the agreement between the received (coded) strip and the checksum, as long as the fingerprints of the coded strip match the encoding of the~$K$ fingerprints in the checksum; this is guaranteed with high probability, as shown~\cite[Theorem~3.4]{AVID-FP}. 

In this respect, each node can assure that the received coded strip is consistent with ones received by others by~\emph{reaching a consensus on the checksum} which is in agreement with all coded strips. 
Treating checksums as requests, we can employ BFT SMR protocols that allow correct nodes to reach a consensus on the total order of them, and hence maintain the consistency of strips at any epoch~$t$. 
Specifically, we adopt HotStuff~\cite{Hotstuff}, a leader-based BFT SMR protocol that works in partial synchrony (see Section~\ref{subsection:smr}).
We define a~\emph{header} of a block~$\bfB$ as a concatenation of checksums computed from the matrix~$\bfB$ and its transpose~$\bfB^\intercal$, i.e.,
$$
\Var{header} = (\Var{cksH},\Var{cksV}).
$$
\new{HotStuff allows nodes to reach a consensus on a chain of header.
HotStuff always ensure~\emph{safety} given bounded number of faulty nodes ($N\geq 3f+1$). 
That is, no two correct node should accept conflicting headers; by conflicting we mean the chain led by neither one extends the chain led by the other. 
Hence, correct nodes will never accept different headers at any epoch~$t$.
When the system becomes synchronous, HotStuff provides the~\emph{liveness} property, such that the consensus on headers will be reached when the leader is correct. 
As discussed earlier, such a consensus on a header guarantees the consistency of coded fragment generated from the corresponding block.}

\new{
For clarity, the lines in Algorithm~\ref{alg:cc1} and Algorithm~\ref{alg:cc2} are color coded. 
The pseudocode of HotStuff is provided in~\emph{black}.
The \blue{blue} lines concern the distribution and verification of coded strips, maintaining validity.
The \orange{orange} lines and \green{green} lines maintain consistency and homology, respectively. 
We will elaborate the colored lines in sequel.
In particular, we argue that our add-ons do not affect the safety and liveness property of HotStuff algorithm.
}

\begin{spacing}{1}
\begin{algorithm}

\caption{Utilities}
\label{alg:utilities}
\begin{multicols}{2}
\begin{algorithmic}[1]

\Function{\FUNC{Msg}}{$\type,\header,\qc,\payload$}
\State $\itm.\type=\type$
\State $\itm.\viewNumber=\curView$
\State $\itm.\header=\header$
\State $\itm.\qc=\qc$
\State $\itm.\payload=\payload$
\State \Return \itm
\EndFunction

\Function{\FUNC{Header}}{$\prev,\checksums$}
\State $\header.\prev=\prev$
\State $\header.\checksums=\checksums$
\EndFunction
\Statex \CommentX{instantiate QC from a set of messages}
\Function{\FUNC{QC}}{$\cM$}
\State $\qc.\type\shortleftarrow \itm.\type:\itm\in\cM$
\State $\qc.\viewNumber\shortleftarrow \itm.\viewNumber:\itm\in\cM$
\State $\qc.\header\shortleftarrow \itm.\header:\itm\in\cM$
\State $\qc.\signature\leftarrow \Var{tcombine}(\langle \qc.\type,\qc.\viewNumber,$
\Statex $\hfill\{\qc.\header\rangle,\itm.\partialSig\mid\itm\in\cM\})$
\EndFunction

\Function{\FUNC{matchingMsg}}{$\itm,\itt,\itv$}
\State \Return $(\itm.\type=\itt)\wedge (\itm.\viewNumber=\itv)$
\EndFunction

\Function{\FUNC{matchingQC}}{$\qc,\itt,\itv$}
\State \Return $(\qc.\type=\itt)\wedge (\qc.\viewNumber=\itv)\wedge$
\Statex $\Var{tverify}(\sig{\qc.\type,\qc.\viewNumber,\qc.\header},\qc.\signature)$
\EndFunction

\Function{\FUNC{Encode}}{$G,\bfm$}\Comment{$\bfm$:length-$K$ vector}
\State \Return $\bfm\cdot G$
\EndFunction

\Function{\FUNC{EncodeRow}}{$G,\bfM$}
\State \Return $ G^\intercal \cdot \bfM$
\EndFunction

\Function{\FUNC{Checksum}}{$G,\bfM$}
\State $\coded{\bfM} \shortleftarrow \FUNC{EncodeRow}(G,\bfM)$ \Comment{$\coded{\bfM}$:~$N\times K$ matrix}
\For{$k=1$ to~$N$}
\State  $\cks.\Var{CC}[i]\shortleftarrow\Var{hash}(\coded{\bfm}_{i.*})$\Comment{ the~$i$-th row of~$\coded{\bfM}$}
\EndFor
\State $\itr~\shortleftarrow \Var{select}(\cks.\Var{CC})$

\For{$k=1$ to~$K$}
\State $\Var{cks.FP}[k] \shortleftarrow \Var{fp}(r, \bfm_{k.*})$\Comment{ the~$k$-th row of~$\bfM$}
\EndFor
\State \Return $\cks,\coded{\bfM}$
\EndFunction
\Statex \CommentX{check if checksum agrees with coded fragment}
\Function{\FUNC{agree}}{$\Var{checksum},\Var{fragment},i$}
\State $\ith \shortleftarrow \Var{hash}(\fragment)$~$\itf \shortleftarrow \Var{fp}(\Var{select}(\cks.\Var{CC}), \fragment)$
\State \Return{$(\ith=\cks.\Var{CC}[i])\wedge(\itf=\FUNC{Encode}(G_\cL,\cks.\Var{FP})[i]$)}

\EndFunction

\Function{\FUNC{safeHeader}}{$\header,\qc$}\label{func:safeHeader2}
\State \Return $(\header~\mbox{extends from~}\lockedQC.\header)\vee$
\State \hfill$(\qc.\viewNumber>\lockedQC.\viewNumber) $\label{line:safetyAndLiveness}
\EndFunction




\Function{\FUNC{signEach}}{$\bfm,i$}\Comment{$\bfm$: length-$N$vector}
\For{$j=1$ to~$N$} $\Var{result}[j]\shortleftarrow \sig{\bfm[j]}_{\sigma_i}$\EndFor
\State \Return \Var{result}
\EndFunction

\end{algorithmic}
\end{multicols}
\end{algorithm}
\end{spacing}

\begin{spacing}{1}
\begin{algorithm}
\caption{Coded Consensus Part 1}\label{alg:cc1}
\begin{algorithmic}[1]
\Statex \red{$\triangleright~\FUNC{prepare phase}$}
\As{a leader}\Comment{$i=\FUNC{Leader}(curView)$}

    \State $\bfB\shortleftarrow \text{collect a block of transactions}$\label{line:block}
    \State $\orange{(\Var{cksH},\coded{\bfB})\shortleftarrow\FUNC{Checksum}(G_\cL,\bfB),(\Var{cksV},\coded{\bfB^\intercal})\shortleftarrow \FUNC{Checksum}(G_\cL,\bfB^\intercal)}$
    
    \Upon{$N-f$ \FUNC{new-view} messages:} $\cM\shortleftarrow \{\itm\mid \FUNC{matchingMsg}(\itm, {\FUNC{new-view}},\Var{curView}-1)\}$    \EndUpon    
    
    \State $\highQC\shortleftarrow( \underset{m\in\cM}{\arg\max}\{\itm.\qc.\viewNumber\}).\qc$ \Comment{QC with highest view number }
    \State $\header\shortleftarrow\FUNC{Header}(\Var{highQC.header}, [\Var{cksH},\Var{cksV}])$
    
    \For{$i=1$ to~$N$}
    \State \blue{$\payload\shortleftarrow[\coded{\bfh}_{i},\coded{\bfv}_{i}]$} \Comment{the~$i$-th rows of~$\coded{\bfB}$ and~$\coded{\bfB^\intercal}$}\label{line:payloadStrips}  
    \State send~$\FUNC{Msg}(\FUNC{prepare},\header,\highQC,\payload)$~to node~$i$
    \EndFor
    
\EndAs

\As{node~$i$}

    \Upon{message~\itm~from \FUNC{Leader}(\curView):}~$\itm\shortleftarrow\FUNC{matchingMsg}(\itm,\FUNC{prepare},\Var{curView})$    \EndUpon
    
    \If{$(\itm.header~\mbox{extends from~}\itm.qc.header)\wedge(\FUNC{safeHeader}(\itm.\header,\header.\qc))$}
    \If{$\orange{\neg(\FUNC{Agree}(\header.\checksum[1],\itm.\payload[1],i)}\wedge\orange{\FUNC{Agree}(\header.\checksum[2],\itm.\payload[2],i)}$}\label{line:consistencyPredicate}
    \State \blue{$\Var{codedResults}\shortleftarrow F(\coded{\bfV}_i,\coded{\bfh}_i)$}\label{line:codedResults}
    
        \State \green{$w_{i,*}\shortleftarrow\FUNC{Encode}(G_\cL,\Var{m.payload}[1])$},~\green{$\px_i\shortleftarrow \FUNC{signEach}(w_{i,*})$}\label{line:signatureVectorX}
        \State \green{$u_{i,*}\shortleftarrow\FUNC{Encode}(G_\cL,\Var{m.payload}[2])$}\label{line:signatureVectorY}        
        \State \blue{$\textit{results}\shortleftarrow\FUNC{Encode}(G_\cL,codedResults)$},~\blue{$\textit{sigResults}\shortleftarrow\FUNC{signEach}(\FUNC{Encode}(G_\cL,codedResults))$}\label{line:encodeResults}\label{line:signResults}
    \State $\Var{ack}\shortleftarrow\FUNC{Msg}(\FUNC{prepare},\itm.\header,\bot,\payload\shortleftarrow(\green{\px_i},\blue{\textit{results}}, \blue{\textit{sigResults}})$
    \State $\ack.\partialSig\shortleftarrow \sig{\FUNC{prepare},\curView,\itm.\header}_{\pi,i}$\label{line:ackPrep}
    \State send \ack~to \FUNC{Leader}(curView)\Comment{acknowledge the \FUNC{prepare} message}\label{line:ackPrepEnd}
    \EndIf
    \EndIf
\EndAs

\Statex \red{$\triangleright~\FUNC{pre-commit phase}$}
\As{a leader}
    \Upon{\ack's~on~\FUNC{prepare} from a quorum~$\cI$:}~$\cA\shortleftarrow  \{\ack\mid \FUNC{matchingMsg}(\ack,  {\FUNC{prepare}},\Var{curView})\}$    \EndUpon    
    \State $\prepareQC\shortleftarrow QC(\cA)$    

    \For{$i=1$ to~$N$}
    \State \green{$\payload[1]\shortleftarrow\FUNC{Column}(\{\ack.\payload[1]\mid\ack\in\cA\},i)$}\label{line:sendColumnX}
    \State \blue{$\payload[3]\shortleftarrow\FUNC{Column}(\{\ack.\payload[3]\mid \ack\in\cA\},i)$}\label{line:sendResultColumn}
    \State \blue{$\payload[4]\shortleftarrow\FUNC{Column}(\{\ack.\payload[4]\mid \ack\in\cA\},i)$}\label{line:sendSigResultColumn}
    \State $\itm\shortleftarrow\FUNC{Msg}(\FUNC{pre-commit},\bot,\prepareQC,\payload),~\green{\itm.\textit{quorumIdentifier}\shortleftarrow\FUNC{QI}(\cI)}$\label{line:sendQI}
    \State send~\itm~to node~$i$
    \EndFor
\EndAs

\As{node~$i$}
    \Upon{message~\itm~from \FUNC{Leader}(\curView)}:~$\itm\shortleftarrow\FUNC{matchingQC}(\itm.qc,\FUNC{pre-commit},\Var{curView})$\EndUpon
    \State {$\prepareQC\shortleftarrow\itm.\qc$}
    \If{$\green{\FUNC{verifySig}(\itm.\payload[1]),u_{i,*})}\wedge\blue{\FUNC{verifySig}(\itm.\payload[2], \itm.\payload[3])}$}   \label{line:verifySig}

    \State \blue{$\Var{decoded}\shortleftarrow \FUNC{Decode}(\itm.\payload[3])$\label{line:decodeTinyBlock}}
    \State \blue{$\Var{binaryResults}\shortleftarrow\FUNC{Binary}(\Var{decoded})$}\label{line:binaryResults}
    \State \blue{$\payload\shortleftarrow\FUNC{partialIndicator}(\Var{binaryResults})$}\label{line:partialIndicator}
     
    \State $\Var{ack}\shortleftarrow\FUNC{Msg}(\FUNC{pre-commit},\itm.\qc.\header,\bot,\payload)$
    \State $\ack.\partialSig\shortleftarrow \sig{ \FUNC{pre-commit},\curView,\itm.\qc.\header}_{\pi,i},~\green{\ack.\Var{partialSigQI}\shortleftarrow\sig{\Var{QI}}_{\pi,i} }$\label{line:partialSigOnQI}

    \State send \ack~to \FUNC{Leader}(curView)  \Comment{acknowledge the \FUNC{pre-commit} message}   
    
    \EndIf
    
\EndAs

\end{algorithmic}
\end{algorithm}
\end{spacing}

HotStuff runs in consecutive~\emph{views} associated with increasing integer view numbers. In each view there is a designated node~\FUNC{Leader}(\viewNumber) that proposes new headers and distribute coded strips.
In order to append a header to the chain, the leader must collect partial signatures on its proposal from a quorum of nodes in each of three phases, namely~\FUNC{prepare}, \FUNC{pre-commit} and~\FUNC{commit}.
The partial signatures are generated with a~$(N-f,N)$-threshold signature scheme~$\pi$.

A new leader must collect~\FUNC{new-view} messages from a quorum of nodes; a correct node sends the~\FUNC{new-view} message, alongside a valid~$\prepareQC$ (define next) with highest view that it has received, to the leader of the next view if it believes the current one fails (line~\ref{line:newView}, Algorithm~\ref{alg:cc2}). 
The new leader chooses the one, called~$\highQC$, with the highest view number within all received~$\prepareQC$s.
It creates and extends the new header (i.e., containing the hash value of another header) from the header contained in~$\highQC$.
If the leader is an incumbent one, it extends the header from its last proposed header.

Incumbent or not, the leader sends the newly created header to every node~$i$ piggybacked with the corresponding coded strips~$\coded{\bfh}_i$ and~$\coded{\bfv}_i$~(\ralg{\ref{line:payloadStrips}}{\ref{alg:cc1}}) generated from the block~$\bfB$ (line~\ref{line:block}, Algorithm~\ref{alg:cc1}). 
Upon receiving the~\FUNC{prepare} message from the leader, the node~$i$ runs the function~\FUNC{safeHeader} in Algorithm~\ref{alg:utilities} which compares the newly received header with the header it has locked on (i.e. the header contained in the~$\precommitQC$ with highest view number it has received, called~$\lockedQC$, which will be defined next).
The new header is considered valid if it extends from the locked header, or extends from a header in a~$\prepareQC$ with a higher view number than the~$\lockedQC$ (line \ref{line:safetyAndLiveness}, Algorithm~\ref{alg:utilities}). 
Such a check guarantees both safety~\cite[Theorem~2]{Hotstuff} and liveness~\cite[Theorem~4]{Hotstuff}.

Further, a valid checksum in the header must agree with the coded strip, and we implement the checking process in the function~\FUNC{agree} in Algorithm~\ref{alg:utilities}. 
In particular, we add another predicate in line~\ref{line:consistencyPredicate} of Algorithm~\ref{alg:cc1} to verify such agreements.

If the aforementioned two predicates both return true, node~$i$ responds with a partial signature~$\sig{\FUNC{prepare},\curView,\header}_{\pi,i}$ acknowledging the header from the leader of the current view (line~\ref{line:ackPrep}, Algorithm~\ref{alg:cc1}).
The leader then enters the~\FUNC{pre-commit} phase, and instantiates a $\prepareQC$ (where \Var{QC} stands for~\emph{quorum certificate}) from the replies using the constructor function~$\FUNC{QC}$ in Algorithm~\ref{alg:utilities}. The~$\prepareQC$ contains a valid signature~$\sig{\FUNC{prepare},\curView,\header}_{\pi}$, showing a quorum of~$N-f$ nodes has acknowledged the~\FUNC{prepare} message from the current leader.

The leader broadcasts~$\prepareQC$. Every node verifies the signature~$\sig{\FUNC{prepare},\curView,\header}_{\pi}$ using the function~\FUNC{matchingQC} in Algortithm~\ref{alg:utilities}.
After that, node~$i$ replies with a partial signature $\sig{\FUNC{pre-commit},\curView,\header}_{\pi,i}$. 
From the replies the leader creates a~$\precommitQC$, and broadcast it in the~\FUNC{commit} phase. Similarly, from the replies the leader creates~$\commitQC$; nodes only link the coded outgoing strip~$\coded{\bfv}_i$ to the local chain after receiving the~$\commitQC$.

\begin{spacing}{1}
\begin{algorithm}
\caption{Coded Consensus Part 2}\label{alg:cc2}
\begin{algorithmic}[1]

\Statex \red{$\triangleright~\FUNC{commit phase}$}
\As{a leader}
    \Upon{$(N-f)$ \ack's~on~\FUNC{pre-commit}:}~$\cA\shortleftarrow \{\ack\mid \FUNC{matchingMsg}(\ack,  {\FUNC{pre-commit}},\Var{curView})\}$    \EndUpon
    \State \blue{$\payload\shortleftarrow\FUNC{mergeIndicators}(\cA)$}\label{line:mergeIndicators}
   \State $\precommitQC\shortleftarrow QC(\cA)$
    \State $\itm\shortleftarrow\FUNC{Msg}(\FUNC{commit},\bot,precommitQC,\payload)$
        \State \green{$\itm.\Var{signatureQI}\shortleftarrow\Var{tcombine}(\sig{QI},\{\ack.\Var{partialSigQI}\mid \ack\in\cA\})$}\label{line:combineSigQI}
    \State broadcast \itm\label{line:broadcastWithSigQI}
\EndAs

\As{node~$i$}
    \Upon{message~\itm~from~\FUNC{Leader}(\curView)}:~$\itm\shortleftarrow\FUNC{matchingQC}(\itm,\FUNC{commit},\Var{curView})$  \EndUpon
    \If{\green{$\Var{tverify}(\sig{QI(\cI)},\itm.\Var{signatureQI})$}}\label{line:QIPredicate}
        \State $\lockedQC\shortleftarrow\itm.\qc$
        \State \blue{$\coded{\bfv}_i\shortleftarrow \FUNC{update}(\coded{\bfv}_i,\itm.\payload)$\Comment{update coded incoming strip using~$g$}}\label{line:updateStrip}
        \State $\Var{ack}\shortleftarrow\FUNC{Msg}(\FUNC{commit},\itm.\qc.\header,\bot,\bot),\ack.\partialSig\shortleftarrow \sig{ \FUNC{commit},\curView,\itm.\qc.\header}_{\pi,i}$
        \State send \ack~to \FUNC{Leader}(curView) \Comment{acknowledge the \FUNC{pre-commit} message}      
    \EndIf
\EndAs

\Statex \red{$\triangleright~\FUNC{decide phase}$}
\As{a leader}
    \Upon{$(N-f)$ \ack's~on~\FUNC{commit}:} $\cA\shortleftarrow \{\ack\mid \FUNC{matchingMsg}(\ack,  {\FUNC{commit}},\Var{curView})\}$    \EndUpon    
    \State $\commitQC\shortleftarrow QC(\cA)$

    \State broadcast $\FUNC{Msg}(\FUNC{decide},\bot,commitQC,\bot)$
\EndAs

\As{node~$i$}
    \Upon{message~\itm~from~\FUNC{Leader}(\curView)}:~$\itm\shortleftarrow\FUNC{matchingQC}(\itm,\FUNC{commit},\Var{curView})$\EndUpon
    \State append~$\coded{\bfv}_i$ to local chain

\EndAs

\Statex \red{$\triangleright~\FUNC{NextView Interrupt}$}
\State send~$\FUNC{Msg}(\FUNC{new-view},\bot,\prepareQC,\bot)$ to~$\FUNC{Leader}(\curView+1)$~\label{line:newView}

\end{algorithmic}
\end{algorithm}
\end{spacing}

\subsection{Maintaining Homology (Condition~\ref{condition:homology})}\label{subsection:maintainingHomology}

Although~\cite{Hotstuff} allows nodes to reach a consensus on the chain of headers, which guarantees the consistency property of strips, the homology problem remains. 
With a Byzantine leader, even though a correct node~$i$ may obtain a consistent coded outgoing strip~$\coded{\bfh}_i=(G_\cL)_i^\intercal\bfB$ and consistent coded incoming strip~$\coded{\bfv}_i=(G_\cL)_i^\intercal(\bfB')^\intercal$, they might correspond to different blocks~$\bfB\neq\bfB'$. 
To solve this problem, we integrate the following design (in \green{green}) with Hotstuff's three-phase protocol to maintain homology.

Upon receiving the coded outgoing strip~$\coded{\bfh}_i=(G_\cL)_i^\intercal\bfB$ from the leader in the~\FUNC{prepare} phase, node~$i$ multiplies it from the right with~$G_\cL$, creating a length-$N$ vector~$w_{i,*}$, which equals to the~$i$-th row of the matrix~$\bfW=G_\cL^\intercal\bfB G_\cL$ (line~\ref{line:signatureVectorX}, Algorithm~\ref{alg:cc1}). 
Similarly, it creates a vector~$u_{i,*}=(G_\cL)_i^\intercal\coded{\bfv}_i$ 
as the~$i$-th row of~$\bfU=(G_\cL)^\intercal(\bfB')^\intercal G_\cL$ (line~\ref{line:signatureVectorY}, Algorithm~\ref{alg:cc1}). 

Node~$i$ then defines a length-$N$ \emph{signature vector}~$\px_i$, whose~$j$-th entry
stores its digital signature (not to be confused with partial signature) on~$\sig{w_{i,j},j}$. 
Formally, we have 
\begin{equation}
        \px_i[j]=\sig{w_{i,j},j}_{\sigma_i},~\text{for }~j\in[N].
\end{equation}

Node~$i$ sends back~$\px_i$ in the acknowledgement of the~\FUNC{prepare} message received from the leader.
For every received~$\px_i$, the leader first verifies if~$\px_i[j]$ is indeed a valid signature on~$w_{i,j}$, for each~$j\in[N]$.
This step is omitted in the pseudocode for clarity, and the leader ignores messages that fail the verification.
After collecting such vectors from nodes in a quorum~$\cI$ of size~$|\cI|=N-f$, the leader stacks the~$\px_i$'s in an~$(N-f)\times N$ matrix ordered by the indices of nodes. 
It sends the~$j$-th column of the resulting matrix to every node~$j\in[N]$ in the~\FUNC{pre-commit} message together with a~\emph{quorum identifier}~$QI(\cI)$ that specifies the members of~$\cI$.\footnote{Since there are~$\binom{N}{f}$ possible quorums,
$\log\binom{N}{f}<\log[\sum_{f=1}^N\binom{N}{f}]=\log 2^N= N$
bits suffice to uniquely present either of them; this is negligible in size compared to the~$N-2f$ digital signatures sent along with it.} (\ralg{\ref{line:sendQI}}{\ref{alg:cc1}}).

Upon receiving the~\FUNC{pre-commit} message from the leader, node~$j$ learns the members of~$\cI$ from the quorum identifier~$QI(\cI)$.
Meanwhile, node~$j$ receives~$\sig{w_{i,j}}_{\sigma_i}$, for every~$i\in\cI$, and verifies if the received~$\sig{w_{i,j},j}_{\sigma_i}$ is a valid signature on~$\sig{u_{j,i},j}$ (line~\ref{line:verifySig}, Algorithm~\ref{alg:cc1}). 
The process is encapsulated in function~\FUNC{verifySig}, whose simple implementation (see above) is omitted for brevity.
If the verification passes, node~$j$ creates partial signature~$\sig{QI(\cI)}_{\pi,i}$ on the quorum identifier and sends it back to the leader as an acknowledgement of the~\FUNC{pre-commit} message (\ralg{\ref{line:partialSigOnQI}}{\ref{alg:cc1}}). 

The leader verifies the received partial signature; this verification is omitted in the pseudocode.
Upon receiving acknowledgements from a quorum~$\cJ$ of nodes, the leader combines partial signatures and broadcasts a~\FUNC{commit} message with a valid signature~$\sig{QI(\cI)}_{\pi}$  (\ralgx{\ref{line:combineSigQI}}{\ref{line:broadcastWithSigQI}}{\ref{alg:cc2}}). 
Nodes can be convinced that~$\bfB=\bfB'$ after verifying~$\sig{QI(\cI)}_{\pi}$ in the~\FUNC{commit} message, due to the following lemma.

\begin{lemma}\label{lemma:validSigQI}
A valid signature~$\langle QI(\cI)\rangle_\pi$ implies~$\bfB=\bfB'$.
\end{lemma}
\begin{proof}

The signature~$\langle QI(\cI)\rangle_\pi$ reveals the existence of a quorum~$\cJ$ such that for every correct node~$j\in\cJ$ and every correct node~$i\in\cI$, we have
$$
w_{i,j}=(G_\cL)^\intercal_i\bfB(G_\cL)_j=(G_\cL)^\intercal_j(\bfB')^\intercal(G_\cL)_i=u_{j,i}.
$$
Since we assume that~$N= (K-1)d+3f+1$, the quorums~$\cI$ and~$\cJ$ intersect on at least
\begin{align*}
    2(N-f)-N&=(K-1)d+f+1 -(K-1)+(K-1) =(K-1)(d-1)+f+K\ge f+K
\end{align*}
nodes, which contains at least~$K$ correct ones.

Let $\cK$ be a set containing these~$K$ correct nodes, and let~$G_{\cK}$ be a~$K\times K$ matrix containing the corresponding~$K$ columns of the Lagrange matrix~$G_\cL$. 
Since~$g_k^\intercal\bfB g_{k'}=g_{k'}^\intercal(\bfB')^\intercal g_k$ for every~$k,k'\in[K]$, it follows that~$G_\cK^\intercal\bfB G_\cK=G_\cK^\intercal\bfB' G_\cK$. 
By the MDS property of~$G_\cL$, the matrix~$G_\cK$ is invertible, and hence~$\bfB=\bfB'$.
\end{proof}




\subsection{Maintaining Validity (Condition~\ref{condition:validity})}\label{subsection:maintainingValidity}

So far, we have developed mechanisms that maintain homology and consistency. 
Together, every correct node~$i$ is performing verification on the coded outgoing strip~$\coded{\bfh}_i=(G_\cL)_i^\intercal \bfB$ and appending the coded incoming strip~$\coded{\bfv}_i=(G_\cL)_i^\intercal \bfB^\intercal$. 
We now present a communication-efficient scheme that employs coded computation to guarantee validity, such that no invalid transactions in the block can be appended to the blockchain. 

\new{
Specifically, we weave a mechanism into the existing protocol. 
It allows nodes to securely obtain the indicator vector~$g\in\{0,1\}^{QK}$.
Note that each of the entries of~$g$ is associated with a coded transaction in every coded incoming strip. 
A coded transaction should be zeroed-out if the corresponding entry is~$1$ (see Section~\ref{subsection:codedComputation}).
}


Recall that the degree of the polynomial verification function~$\bfF(z)$ is~$(K-1)d$, and hence it is uniquely defined by evaluations at any~${L=(K-1)d+1}$ distinct points. That is, for any distinct~$\beta_1,\ldots,\beta_L$, one may represent~$\bfF(z)$ as a linear combination of~$L$ Lagrange basis polynomials~$\Psi_1(z),\ldots,\Psi_L(z)$, i.e.,
\begin{equation*}
    \bfF(z)=\sum_{\ell\in[L]}\bfF(\beta_i)\Psi_\ell(z),~\text{where}~\Psi_\ell(z)=\prod_{l,\ell\in[L], l\neq \ell}\frac{z-\beta_l}{\beta_\ell-\beta_l}.
\end{equation*}



As a result, the coded outgoing result strips~$\bfF(\alpha_1),\ldots,\bfF(\alpha_N)$ can be represented as
\begin{equation}\label{eq:codedResult}
    \begin{bmatrix} \coded{\bfe}_{1}\\\vdots\\\coded{\bfe}_{N}\end{bmatrix}\overset{\eqref{equation:FeCodedUncoded}}{=}
    \begin{bmatrix}{\bfF(\alpha_1)}\\ \vdots\\{\bfF(\alpha_N)}\end{bmatrix}=
    G_{\cF,\alpha}^\intercal\cdot
    \begin{bmatrix} \bfF(\beta_1)\\\vdots\\\bfF(\beta_L)\end{bmatrix} =
    \begin{bmatrix}
\Psi_1(\alpha_1) & \Psi_1(\alpha_2) & \ldots &\Psi_1(\alpha_N) \\
\Psi_2(\alpha_1) & \Psi_2(\alpha_2) & \ldots & \Psi_2(\alpha_N) \\
\vdots &\vdots&\ddots&\vdots\\
\Psi_{L}(\alpha_1) & \Psi_{L}(\alpha_2) & \ldots & \Psi_{L}(\alpha_N) \\
\end{bmatrix}^\intercal\cdot 
\begin{bmatrix} \bfF(\beta_1)\\\vdots\\\bfF(\beta_L)\end{bmatrix},
\end{equation}
and the (uncoded) outgoing result strips can be represented as
\begin{equation}\label{eq:decodeResult}
    \begin{bmatrix}{\bfe}_{1}\\ \vdots\\{\bfe}_{K}\end{bmatrix}\overset{\eqref{equation:FeCodedUncoded}}{=}
    \begin{bmatrix}{\bfF(\omega_1)}\\ \vdots\\{\bfF(\omega_K)}\end{bmatrix}
    = G_{\cF,\omega}^\intercal\cdot
    \begin{bmatrix} \bfF(\beta_1)\\\vdots\\\bfF(\beta_L)\end{bmatrix}
    =\begin{bmatrix}
\Psi_1(\omega_1) & \Psi_1(\omega_2) & \ldots &\Psi_1(\omega_K) \\
\Psi_2(\omega_1) & \Psi_2(\omega_2) & \ldots & \Psi_2(\omega_K) \\
\vdots &\vdots&\ddots&\vdots\\
\Psi_{L}(\omega_1) & \Psi_{L}(\omega_2) & \ldots & \Psi_{L}(\omega_K) \\
\end{bmatrix}^\intercal\cdot 
\begin{bmatrix} \bfF(\beta_1)\\\vdots\\\bfF(\beta_L)\end{bmatrix}.
\end{equation}

Upon receiving the message from the leader in the~\FUNC{prepare} phase, node~$i$ computes the verification function~$\bfF$ and obtains its coded outgoing result strip~$\coded{\bfe}_{i}$ (\ralg{\ref{line:codedResults}}{\ref{alg:cc1}}). 
Node~$i$ multiplies it from the right with the Lagrange matrix~$G_\cL\in\bbF_q^{K\times N}$, and obtains
$
    c_{i,*}=(\coded{\bfe}_{i,1},\ldots,\coded{\bfe}_{i,K})\cdot G_\cL,
$
which equals to the~$i$-th row of the matrix
\begin{equation}
    \bfC = G_{\cF,\alpha}^\intercal \cdot 
    \begin{bmatrix}\bfF(\beta_1)^\intercal,\ldots,\bfF(\beta_L)^\intercal \end{bmatrix}^\intercal
    \cdot G_\cL.
\end{equation}

In the acknowledgment of the~\FUNC{prepare} message, node~$i$ replies the leader with~$c_{i,*}$ with its signatures on each entry (\ralg{\ref{line:encodeResults}}{\ref{alg:cc1}}). 
\new{The leader verifies if the signatures matches~$c_{i,*}$; this step is omitted in the pseudocode for clarity}.
Upon receiving a quorum of~$N-f$ such vectors, the leader stacks them on top of each other to form a~$(N-f)\times N$ matrix (which is a submatrix of~$\bfC$), and sends the~$j$-th column to every node~$j\in[N]$ in the~\FUNC{pre-commit} phase (\ralg{\ref{line:sendResultColumn}}{\ref{alg:cc1}}).
Note that the~$j$-th column of~$\bfC$ is the encoding of the~$j$-th column of matrix~$[\bfF(\beta_1)^\intercal,\ldots,\bfF(\beta_L)^\intercal]^\intercal\cdot G_\cL$ using the generator matrix~$G_{\cF,\alpha}$, which generates a Lagrange code of length~$N$ and dimension~$L$. 
As a result, every node can perform Reed-Solomon decoding after verifying the signature of each entry (line~\ref{line:verifySig}, Algorithm~\ref{alg:cc1}), and obtain the~$j$-th column of matrix~$[\bfF(\beta_1)^\intercal,\ldots,\bfF(\beta_L)^\intercal]^\intercal\cdot G_\cL$ (line~\ref{line:decodeTinyBlock}, Algorithm~\ref{alg:cc1}). 
The decoding is given in the function~\FUNC{Decode} which calls a Reed-Solomon decoder.
Since we have~$N\ge (K-1)d+3f+1$, decoding from~$N-f$ elements will be successful since there are at most~$f$ Byzantine nodes.


By left multiplying the decoded column with~$G^\intercal_{\cF,\omega}$, every correct node~$i$ obtains the~$j$-th column of the matrix~$
   G^\intercal_{\cF,\omega}\cdot
   \begin{bmatrix}\bfF(\beta_1)^\intercal,\ldots,\bfF(\beta_L)^\intercal \end{bmatrix}^\intercal
   \cdot G_\cL$.
This vector equals to the~$j$-th column of~$\begin{bmatrix}{\bfe}_{1}^\intercal, \ldots,{\bfe}_{K}^\intercal\end{bmatrix}^\intercal\cdot G_\cL$ 
by Equation~\eqref{eq:decodeResult}, which further equals to the~$j$-th coded incoming result strip by Equation~\eqref{eq:result block} and Equation~\eqref{eq:incoming result strip}, i.e.,~$\coded{\bfs}_j= \bfR\cdot (G_\cL)_j$.

Recall that the result block~$\bfR$ is a~$K\times K$ matrix whose each element~$r_{k,k'}$ stores the verification results of the~$Q$ transactions in the tiny block~$b_{k,k'}$. 
Therefore, a coded incoming results strip~$\coded{\bfs}_i$ contains~$K$ coded tiny result blocks; the~$k$-th one is a linear combination of~$r_{k,1},\ldots,r_{k,K}$ defined by~$(G_\cL)_k$. 
Hence, for~$l\in[Q]$, the~$l$-th entry in the~$k$-th coded tiny result block is a linear combination of verification results of transactions in~$\cS_{k,l}$; a set containing every~$l$-th transaction in~$b_{k,1},\ldots,b_{k,K}$.
If the entry is not a zero vector, it suggests that at least one of these verification results is not a zero vector, which further suggests at least one transaction in~$\cS_{k,l}$ is invalid. 
On the other hand, if the entry is a zero vector, node~$i$~\emph{cannot} conclude the validity of transactions in~$\cS_{k,l}$, as a linear combination of non-zero vectors might be the zero vector.

Recall that in an MDS code of dimension~$K$, every codeword is either the zero codeword, or has at most~$K-1$ zeros. 
In this regard, the~$l$-th entries in the~$k$-th coded tiny result block from all~$\coded{\bfs}_1,\ldots,\coded{\bfs}_N$ form a codedword of an~$[N,K]$ MDS code, and hence contains either all zero vectors, or at most~$K-1$ zero vectors (note that a vector is an element in the codeword). 

The former case implies that each of the~$l$-th transactions in~$b_{k,1},\ldots,b_{k,K}$ passes verification. 
The latter case implies that at least one of theses transactions is invalid, and the~$l$-th coded transaction in the~$k$-th coded tiny block of every coded incoming strip must be set to zero before being appended.
To simplify the problem, every node~$i$ creates a~\emph{binary results} vector~$g_{*,i}$, which is the~$i$-th column of matrix~$\bfG\in\{0,1\}^{QK\times N}$.
Each entry of~$g_{*,i}$ is associated with an entry of~$\coded{\bfs}_i$; it equals to~$0$ if the corresponding entry in~$\coded{\bfs}_i$ is a zero vector, and equals to~$1$ otherwise (line~\ref{line:binaryResults}, Algorithm~\ref{alg:cc1}). This operation is encapsulated in the function~\FUNC{Binary}, whose pseudocode implementation is omitted for its simplicity. 
Note that each row of~$\bfG$ is either all zeros, or contains at most~$K-1$ zeros. Clearly, the indicator vector~$g$ equals to the reduction of all columns of~$\bfG$ with operator bitwise OR.

\newcommand{\gw}{\Var{gw}}

Let~$\lambda$ be a~$(K+f, N)$ threshold signature scheme, and let~$\tau$ be a~$(f+1, N)$ threshold signature scheme. 
Using the binary results vector~$g_{*,i}$, node~$i$ obtains a~\emph{partial indicator} as the output of the function~\FUNC{partialIndicator}~(line~\ref{line:partialIndicator}, Algorithm~\ref{alg:cc1}). 
This function defines a length-$QK$ vector, denoted by~$\gw_i$, such that for every~$\ell\in[QK]$, 
\begin{equation*}
    \gw_i[\ell]=\begin{cases}\sig{\ell, 0 ,\header}_{\lambda,i} &g_{\ell,i}=0\\
   \sig{\ell, 1 ,\header}_{\tau,i} & g_{\ell,i}=1
    \end{cases}.
\end{equation*}

Node~$i$ sends~$\gw_i$ back to the leader in the acknowledgment of the~\FUNC{pre-commit} message.
The leader collects~$\gw_i$'s from a quorum of~$N-f$ nodes and merges them into a length-$QK$ vector~$\gw$ using function~\FUNC{mergeIndicators} (line~\ref{line:mergeIndicators}, Algorithm~\ref{alg:cc2}); the details are given as follows. 

Among the~$\ell$-th entries of the collected vectors~$\{\gw_j\}_{j\in\cJ}$, if there exist~$K+f$ partial signatures endorsing~$0$ (generated by the~$\lambda$ scheme), the leader generates and stores a valid signature~$\sig{\ell, 0 ,\header}_\lambda$ in the~$\ell$-th entry of~$\gw$. 
Otherwise, if there exists~$f+1$ partial signatures endorsing~$1$ (generated by the~$\tau$ scheme), the leader stores a valid signature~$\sig{\ell, 1 ,\header}_\tau$. 
Notice that exactly one of these cases must hold due to the following lemma. 
\new{Note that we implicitly assume that the leader is guaranteed to obtain responses from a quorum~$\cJ$ in the \FUNC{pre-commit} phase; such an assumption will be justified in Theorem~\ref{theorem:liveness} on the liveness property of our scheme.}

\begin{lemma}
Among the~$\ell$-th entries of the collected vectors~$\{\gw_j\}_{j\in\cJ}$ from a quorum of size $|\cJ|=N-f$, the leader is guaranteed to obtain~$K+f$ partial signatures endorsing~$0$, or~$f+1$ partial signatures endorsing~$1$, but not both.
\end{lemma}
\begin{proof}
For any~$\ell\in[QK]$, if the~$\ell$-th row of~$\bfG$ is all-zero, then at least
\begin{align*}
    N-2f&= (K-1)d+f+1\ge (K-1)d+f+1 -(K-1)+(K-1)= (K-1)(d-1)+f+K\ge f+K
\end{align*}
vector~$\gw_i$'s are from correct nodes; they all have zero~$\ell$-th entry and sign using the~$\lambda$ scheme. Meanwhile, there exist at most~$f$~$1$'s, all from the Byzantine nodes.

If the~$\ell$-th row of~$\bfG$ is not all-zero, then the number of nodes (at least~$N-(K-1)$) having~$1$'s must intersect with the quorum on at least~$(N-f)+N-(K-1)-N$ nodes, which equals to
\begin{align*}
    N-f-(K-1)&= (K-1)d+2f+1 -(K-1)= (K-1)(d-1)+2f+1\ge 2f+1
\end{align*}
nodes, out of which at least~$f+1$ are correct; they all endorse~$1$ and sign the entry using the~$\tau$ scheme. Also, there exist at most~$(K-1+f)$~$0$'s, out of which~$K-1$ are from correct nodes, and at most~$f$ are from Byzantine nodes. 
\end{proof}

The leader then broadcasts the vector~$\gw$ to every node in the~\FUNC{commit} phase. Every node~$i$ can learn the indicator vector~$g$ from~$\gw$, i.e., for every~$\ell\in[QK]$,
\begin{equation*}
    g[\ell]=\begin{cases} 0 & \gw[\ell]=\sig{\ell, 0 ,\header}_{\lambda} \\
   1& \gw[\ell]=\sig{\ell, 1 ,\header}_{\tau}
    \end{cases}.
\end{equation*}
It then uses the indicator variable to~``filter out'' invalid transactions in the coded incoming strip~$\coded{\bfv}_i$ (\ralg{\ref{line:updateStrip}}{\ref{alg:cc2}}).

\section{Discussion}\label{section:discussion}

In this section, we discuss the security,~\new{liveness}, and the communication complexity aspects of our design. In particular, we investigate the tradeoff between bit complexity and security level.

\subsection{Security}
The security level of our scheme is reflected by the upper bound of~$f$ compared with~$N$, i.e., the maximum fraction of Byzantine nodes that can be tolerated in the system. The following theorem shows that, for correct verification of transactions,~$f$ depends on the total number of nodes~$N$, the number of shards~$K$, and the degree~$d$ of the verification function. 
\begin{theorem}\label{theorem:computation}
If~$N\geq (K-1)\cdot d+3f+1$, our design provides coded consensus.
\end{theorem}
\begin{proof}

First, HotStuff guarantees safety~\cite[Theorem~2]{Hotstuff}
(see Section~\ref{section:communication} for definitions) of the header chain when~$N\ge3f+1$, which is a weaker assumption than~$N\geq (K-1)d+3f+1$. 
\new{Note that the added mechanisms are irrelevant to the safety property, as no extra conditions on which nodes can accept a header are introduced.}
The property of homomorphic fingerprinting function assures the consistency between the coded fragments received by each node~\cite[Theorem~3.4]{AVID-FP}. Together, consistency is maintained.

Second, as seen in Lemma~\ref{lemma:validSigQI}, our method maintains homology between the coded incoming strips and the coded outgoing strips when~$N\geq (K-1)\cdot d+3f+1$.

Finally, in order to obtain the indicator vector~$g$, every node needs to decode an~$[N,L]$ Reed-Solomon code from~$N-f$ elements in the codeword, where~$L=(K-1)d+1$ (see Section~\ref{subsection:maintainingValidity}). Since~$N-f\geq (K-1) d+2f+1$, the property of Reed-Solomon code guarantees correct decoding in this case. Thus, validity is maintained. \qedhere
\end{proof}

\subsection{Liveness}
\new{
Although the proposed algorithm provides coded consensus, adversaries may conduct a liveness attack, i.e., prevent the system from processing new transactions.
In this section, we show that the proposed algorithm also provides liveness.

\begin{theorem}\label{theorem:liveness}
In the partial synchrony model, the proposed algorithm provides liveness after Global Stabilization Time (GST, see Section~\ref{subsection:smr}).
\end{theorem}
\begin{proof}
As shown in~\cite[Theorem~4]{Hotstuff}, HotStuff provides liveness after GST.
That is, a decision is reached given that there is a bounded duration~$T_f$, in which all correct nodes remain in the same view, and the view-leader is correct.
We show that this property is preserved with the added mechanisms.
Specifically, in our modified algorithm, there are precisely three occasions, one in each phase, in which liveness can be affected: line~\ref{line:consistencyPredicate}, Algorithm~\ref{alg:cc1}, line~\ref{line:verifySig}, Algorithm~\ref{alg:cc1}, and line~\ref{line:QIPredicate}, Algorithm~\ref{alg:cc2}. 
In these occasions, a correct leader might fail to collect sufficiently many responses, and thus liveness might not be guaranteed. 
We show that each of these occasions depends on a Boolean predicate which is guaranteed to be satisfied when the leader is correct, and thus liveness is preserved.


The first predicate (line~\ref{line:consistencyPredicate}, Algorithm~\ref{alg:cc1}) checks if the received header agrees with the received strips.
It is true in every correct node since a correct leader follows the protocol.
Therefore, a correct leader is guaranteed to receive valid responses in the \FUNC{prepare} phase from~$N-f$ nodes. 

For the second predicate (line~\ref{line:verifySig}, Algorithm~\ref{alg:cc1}), given the~$N-f$ valid responses from the \FUNC{prepare} phase, a correct leader is able to construct two~$(N-f)\times N$ matrices. The~$j$-th row of these matrices will make the green-colored function calls in line~\ref{line:verifySig}, Algorithm~\ref{alg:cc1} to return true for node~$j$. 
For the same reason, a correct leader is able to construct an~$(N-f)\times N$ matrix, whose~$j$-th row will make the blue-colored function call in the same line true. 
Therefore, every correct node will respond in the \FUNC{pre-commit} phase, and hence the correct leader will receive responses, each containing a valid partial signature on~$\Var{QI}$, from~$N-f$ nodes.

Finally, the leader is able to generate a valid signature~$\sig{\Var{QI}(\cI)}_\pi$ on the quorum identifier from the partial signatures.
Therefore, the third predicate (line~\ref{line:QIPredicate}, Algorithm~\ref{alg:cc2}) is true as well. 
\end{proof}
}

\subsection{Communication Complexity}

We analyze the communication complexity for the system to process a block~$\bfB$ that contains~$P=QK^2$ transactions, and then compare it to ordinary blockchain designs. The bit complexity of the different stages of our protocol is analyzed next, and sumarized in Table~\ref{table:cc}. Note also that the message complexity is linear thanks to the HotStuff protocol in use.


\begin{table}[!h]
\centering
\begin{tabular}{ccccc}
\hline
\multicolumn{1}{|c|}{}       & \multicolumn{1}{c|}{\textsc{prepare}}       & \multicolumn{1}{c|}{\textsc{pre-commit}}   & \multicolumn{1}{c|}{\textsc{commit}} & \multicolumn{1}{c|}{\textsc{decide}} \\ \hline
\multicolumn{1}{|c|}{Leader} & \multicolumn{1}{c|}{$O(N\log N+dQK \log N)$} & \multicolumn{1}{c|}{$O(NQ\log N)$} & \multicolumn{1}{c|}{$O(QK)$}         & \multicolumn{1}{c|}{$O(1)$}          \\ \hline
\multicolumn{1}{|c|}{Node}   & \multicolumn{1}{c|}{$O(NQ \log N)$}   & \multicolumn{1}{c|}{$O(QK)$}               & \multicolumn{1}{c|}{$O(1)$}          & \multicolumn{1}{c|}{N/A}             \\ \hline
\multicolumn{1}{l}{}         & \multicolumn{1}{l}{}                        & \multicolumn{1}{l}{}                       & \multicolumn{1}{l}{}                 & \multicolumn{1}{l}{}                
\end{tabular}\caption{Bit complexities of a single message from the leader to a node, and from a node to the leader, in each of the stages.}\label{table:cc}
\end{table}

In the~$\FUNC{prepare}$ phase, the leader sends a checksum and two coded fragments to each of the~$N$ nodes. 
A checksum contains~$N$ signatures over~$\bbF_q$, and a coded fragment contains~$\frac{|\bfB|}{K}$ bits. 
Recall that a block~$\bfB$ contains~$QK^2$ transactions, and each contains a lookup table whose size scales logarithmically with the number of transactions in a shard, same as the degree of the polynomial verification function~$d$ (see Section~\ref{subsection:pvf}). 
Further, since the underlying field~$\bbF_q$ must contain at least~$N$ distinct elements, it follows that the size of a field element is~$O(\log N)$ bits.
Together, the size of a block is~$O(dQK^2\log N)$, and the size of a coded strip is~$O(dQK\log N)$.
Note that the leader also broadcast the header, which contains~$2N$ hash values and~$2K$ fingerprints, each has a constant number of field elements.
Therefore, the message from the leader to a single node in this step is~$O(N\log N+dQK\log N)$. 

Also in the~$\FUNC{prepare}$ phase, every node~$i$ sends~$N$ signatures (line~\ref{line:signatureVectorX}, Algorithm~\ref{alg:cc1}), as well as~$N$ coded tiny result blocks~(line~\ref{line:encodeResults}), to the leader.
Recall that every coded tiny result block contains~$Q$ verification results, each is a length-$(C+E)$ vector, where~$C+E$ is the outputs of hash functions and hence constant.
Therefore, each message from a node to the leader in the~\FUNC{prepare} phase has a size of~$O(N+NQ(C+E)\log N)=O(NQ\log N)$.

In the~$\FUNC{pre-commit}$ phase, node~$i$ receives~$(N-2f)$ signatures~(line~\ref{line:sendColumnX}, Algorithm~$\ref{alg:cc1}$) and~$N-2f$ coded tiny result blocks~(line~\ref{line:sendResultColumn}) back from the leader. Therefore, the size of a message from the leader to a node is also~$O(N+NQ\log N)=O(NQ\log N)$.
Next, still in the~\FUNC{pre-commit} phase, every node~$i$ sends a partial indicator vector~(line~\ref{line:partialIndicator}) to the leader, whose size is~$O(QK)$ as it contains~$QK$ partial signatures.
In the~$\FUNC{commit}$ phase, every node receives a length-$QK$ vector of threshold signatures~(line~\ref{line:mergeIndicators}). 
In addition, every message sent to the leader contains a partial signature and hence has a size of~$O(1)$. 
Similarly, every message sent from the leader in the~$\FUNC{decide}$ phase contains a threshold signature~(in~\commitQC), and hence has size~$O(1)$. Together, the bit complexity of our design is as follows.

\begin{corollary}\label{corollary:bitComplexity}
For~$\mu<1/3$, to tolerate~$\mu N$ Byzantine nodes in a system with~$N$ nodes, the overall bit complexity for verifying a block of~$P=K^2Q$ transactions is~$O(\frac{Pd^2\log N}{(1-3\mu)^2})$.
\end{corollary}
\begin{proof}

From Table~\ref{table:cc}, the overall bit complexity is
$$O(N^2\log N+dNQK\log N+N^2Q\log N+NQK+N)=O(N^2Q\log N+dNQK\log N).$$
Taking the maximum possible~$f$ given the parameter restriction in Theorem~\ref{theorem:computation}, we have that~$N= (K-1)\cdot d+3f+1$, and hence for~$n=\frac{(K-1)d}{f}$ we have
\begin{equation}\label{eq:CC}
\frac{N}{K}\approx\frac{N-1}{K-1}=\frac{(3+n)f}{nf/d}=d\left(1+\frac{3}{n}\right),
\end{equation}
Further, since~$
     \frac{N^2Q\log N}{dNQK\log N}=\frac{N}{Kd}
            \approx 1+\frac{3}{n}\geq 1$,
it follows that the overall bit complexity is~$O(N^2Q\log N)$.
As we have~$N= (K-1)d+3f+1$ by Theorem~\ref{theorem:computation}, the system tolerates a fraction~$\mu=\frac{f}{N}=\frac{1}{3+n+1/f}\approx \frac{1}{3+n}$ of Byzantine nodes. We can now express the overall bit complexity as a function of~$\mu$:
\begin{align*}
     O(N^2Q\log N)&=O( (N/K)^2K^2Q\log N)=O(Pd^2(1+3/n)^2\log N)=O\left(P\frac{d^2\log N}{(1-3\mu)^2}\right).\qedhere
\end{align*}
\end{proof}

That is, for a system of~$N$ nodes and the verification function of degree~$d$, the system designer can choose a value for~$\mu$, and the bit complexity for verifying a block scales quadratically with~$d$ and logarithmically with~$N$. 
Note that the degree~$d$ scales logarithmically with the number of transactions on one shard. We hereby rewrite the bit complexity for verifying one block as
$$
O(P\log^2 M(t) \log N),
$$
where~$M(t)$ equals to the number of transactions on one shard at epoch~$t$.

To show the novelty of our design, we define the~\emph{communication gain}~$\cG$ as the ratio between the bit complexity common in ordinary blockchain systems, which require every node to receive every transaction, and the bit complexity of our design; the former leads to an inevitable~$O(NP)$ bit complexity assuming that each transaction requires a constant amount of bits, and a block contains~$P$ transactions.
Specifically, if the system in our design tolerates~$\mu N$ Byzantine nodes, where~$\mu<\frac{1}{3}$, the communication gain is
\begin{equation}\label{eq:gain}
\cG=\frac{NP}{P\frac{d^2\log N}{(1-3\mu)^2}}=\frac{N(1-3\mu)^2}{d^2\log N}.
\end{equation}

It is evident from~\eqref{eq:gain} that the communication gain is significant for any fixed value of~$\mu$ and~$d$. Moreover, increasing the number of nodes in the system while keeping the remaining parameters fixed \emph{improves} the overall communication gain with respect to traditional designs; this is a highly desirable property of blockchain systems.

\subsection{Communication-Security Tradeoff}
By Corollary~\ref{corollary:bitComplexity}, the overall bit complexity is~$O(P\frac{d^2\log N}{(1-3\mu)^2})$, from which a tradeoff between security and communication is evident.
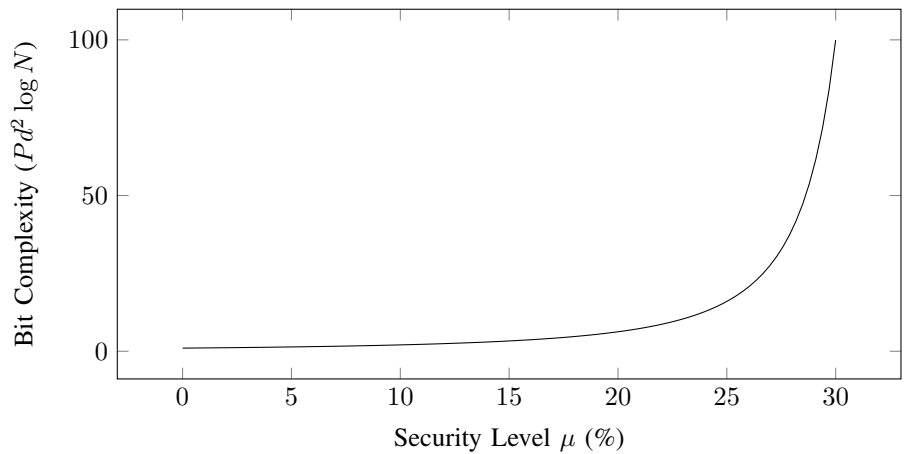
\begin{figure}[h]
\centering
\begin{tikzpicture}
  \begin{axis}[xlabel =Security Level~$\mu$ (\%),
  ylabel = {Bit Complexity~($Pd^2\log N$)},
  width=12cm,
  height=6.5cm
  ]
    \addplot[domain=0:30,samples=100]
    {1/(1-3*(x/100))^2)}; 
  \end{axis}
\end{tikzpicture} 
\caption{An illustration of the tradeoff between the security level~$\mu=\frac{f}{N}$ and communication bit complexity.}\label{figure:tradeoff}
\end{figure}
A lower~$\mu$ value yields low bit complexity, but degrades the security level (since~$\mu=f/N$).
In contrast, a higher~$\mu$ value allows the system to tolerate more Byzantine nodes, but inevitably leads to a higher bit complexity. 
In Figure~\ref{figure:tradeoff} we illustrate the function~$\mu\mapsto \frac{1}{(1-3\mu)^2}$, which describes the tradeoff between~$\mu$ and the bit communication complexity, measured relative to the baseline~$Pd^2\log N$ in Corollary~\ref{corollary:bitComplexity}.

\section{Future Work and Concluding Remarks}\label{section:future}
\new{
This paper focuses on verifying the validity of new transactions, but does not discuss how nodes can learn if an old transaction has already been redeemed. 
Directions for future work include incorporating light nodes, and developing algorithms for them to access raw data by querying a coded distributed system with Byzantine nodes. 
Finally, as this paper adopts a simplified UTXO model, the generalized multi-input multi-output setting is an interesting direction for future research.
In spite of these disadvantages, our work shows that coded computation can alleviate the communication burden in blockchain systems, while maintaining the computations and storage benefits of sharding. }

\end{document}